\documentclass[11pt,authoryear]{elsarticle}

\makeatletter \def\ps@pprintTitle{ \let\@oddhead\@empty \let\@evenhead\@empty \def\@oddfoot{\hfill\thepage} \def\@evenfoot{\thepage\hfill}} \makeatother

\usepackage{amsthm,amssymb,amsmath,mathtools,amsbsy}

\usepackage{graphicx,rotate}

\usepackage{tikz}
\usepackage{subcaption}
\usepackage{float}
\usepackage{dsfont}
\usepackage{hyperref}
\usepackage{bigstrut}
\usepackage{multirow}
\usepackage{setspace}
\usepackage{tikz}
\usepackage{algorithm2e}
\usepackage{ulem}
\usepackage{enumitem}
\setlist[enumerate]{leftmargin=.5in}
\setlist[itemize]{leftmargin=.5in}

\newtheorem{theorem}{Theorem}[section]
\newtheorem{lemma}[theorem]{Lemma}

\newtheorem{proposition}[theorem]{Proposition}

\theoremstyle{definition}

\newtheorem{assumption}[theorem]{\sc Assumption}

\newtheorem{remark}[theorem]{\bf Remark}

\def\D{\mathcal{D}}
\def\A{\mathcal{A}}

\def\mfK{{\mathfrak{K}}}
\def\F{{\mathcal{F}}}
\def\G{{\mathcal{G}}}
\def\K{{\mathcal{K}}}
\def\M{\mathcal{M}}

\def\mcT{\mathcal{T}}

\def\Z{\mathcal{Z}}
\def\bZ{\bm{\Z}}

\def\mfQ{{\mathfrak{Q}}}
\def\mfN{{\mathfrak{N}}}
\def\mfJ{{\mathfrak{J}}}

\def\mcM{{\mathcal{M}}}
\def\mcMbar{\overline{\mcM}}
\def\bmcMbar{\bm{\mcMbar}}
\def\bmcMt{\bm{\tilde{\mcM}}}

\def\mcE{\mathcal{E}}
\def\bmcE{\bm{\mcE}}

\def\bGamma{\bm{\Gamma}}

\def\bpi{\bm{\pi}}
\def\bC{\bm{C}}

\def\bM{\bm{M}}
\def\bY{\bm{Y}}
\def\bbeta{\bm{\beta}}
\def\bbetah{\widehat{\bbeta}}

\def\E{\mathbb{E}}
\def\P{\mathbb{P}}
\def\P{{\mathbb{P}}}
\def\R{\mathbb{R}}
\def\HT{\mathbb{H}^2_T}

\def\LTk{\mathbb{L}^{2,k}_T}
\def\HTk{\mathbb{H}^{2,k}_T}
\def\Ek{\E^\Pk}
\def\EP{\E^\P}
\def\Pk{{\P^k}}
\def\Q{\mathbb{Q}}

\newcommand{\bm}{\boldsymbol}

\def\g{\bm{g}}

\def\bnubar{\bm{\nubar}}
\def\bqbar{\bm{\qbar}}

\def\ba{\bm{a}}
\def\bLambda{\bm{\Lambda}}
\def\blambda{\bm{\lambda}}
\def\bPsi{\bm{\Psi}}
\def\bphi{\bm{\phi}}
\def\bg{\bm{g}}
\def\blambda{\bm{\lambda}}

\def\bnubarN{\bnubar^{\N}}
\def\bnubarkN{\bnubar^{k,\N}}
\def\bG{\bm{G}}

\def\Lambdat{\tilde{\Lambda}}

\def\bgh{\hat{\bg}}
\def\T{\intercal}

\def\Ahat{\widehat{A}}

\def\bAhat{\bm{\Ahat}}

\def\nubar{\overline{\nu}}
\def\qbar{\overline{q}}

\def\bmcEt{\tilde{\bmcE}}

\renewcommand{\hbar}{\overline{h}}

\def\N{(N)}

\def\nuj{\nu^j}
\def\nujst{\nu^{j,\ast}}

\def\nubarN{\nubar^{\N}}
\def\nubarkN{\nubar^{k,\N}}
\def\nubark{\nubar^k}

\def\bnubarst{\bnubar^{\ast}}
\def\nubarkst{\nubar^{k,\ast}}

\def\qbar{\bar{q}}

\def\Hbar{\overline{H}}

\def\mMbar{\overline{\M}}
\def\bgt{\tilde{\bg}}
\def\invmean{\overline{m}}

\def\mbar{\bar{m}}
\def\bmbar{\bm{\mbar}}

\def\kp{k^\prime}

\newcommand{\1}[1]{\mathds{1}_{\left\{ {#1} \right\} }}

\newcommand{\smallqfm}[1]{\left( \begin{smallmatrix} {#1}_t \\ \qj{#1}_t \end{smallmatrix} \right)}

\newcommand{\qj}[1]{q^{j,{#1}}}

\newcommand{\qbark}[1]{\bar{q}^{k,{#1}}}

\newcommand{\OPnrm}[1]{\bigl\lVert {#1} \bigr\lVert_2}
\newcommand{\nrm}[1]{\left\lVert {#1} \right\lVert}

\newcommand{\bnu}{{\nubarN}}

\DeclareMathOperator*{\argmax}{arg\,max}

\newcommand{\mcB}{{\mathcal{B}}}

\newcommand{\hg}{{\widehat{g}}}
\newcommand{\tT}{{t\in[0,T]}}
\newcommand{\tomega}{{\widetilde\omega}} 

\begin{document}

\begin{frontmatter}

\title {
\textbf{Mean-Field Games with Differing Beliefs for Algorithmic Trading}\tnoteref{t1}\\[0.5em]
\textit{Forthcoming in Mathematical Finance}}
\tnotetext[t1]{SJ would like to acknowledge the support of the Natural Sciences and Engineering Research Council of Canada (NSERC), funding reference numbers RGPIN-2018-05705 and RGPAS-2018-522715. Data sharing is not applicable to this article as no new data were created or analyzed in this study.}


\author[author1]{Philippe Casgrain}
\ead{p.casgrain@utoronto.ca}

\author[author1]{Sebastian Jaimungal}
\ead{sebastian.jaimungal@utoronto.ca}
\address[author1] {Department of Statistical Sciences, University of Toronto}

\begin{abstract}
Even when confronted with the same data, agents often disagree on a model of the real-world. Here, we address the question of how interacting heterogenous agents, who disagree on what model the real-world follows, optimize their trading actions. The market has latent factors that drive prices, and agents account for the permanent impact they have on prices. This leads to a large stochastic game, where each agents' performance criteria are computed under a different probability measure. We analyse the mean-field game (MFG) limit of the stochastic game and show that the Nash equilibrium is given by the solution to a non-standard vector-valued forward-backward stochastic differential equation. Under some mild assumptions, we construct the solution in terms of expectations of the filtered states. Furthermore, we prove the MFG strategy forms an $\epsilon$-Nash equilibrium for the finite player game. Lastly, we present a least-squares Monte Carlo based algorithm for computing the equilibria and {show through simulations that increasing disagreement may increase price volatility and trading activity}.
\end{abstract}

\end{frontmatter}

\section{Introduction}

Financial markets are immensely complicated dynamic systems which incorporate the interactions of millions of individuals on a daily basis. Market participants vary immensely, both in terms of their trading objectives and in their beliefs on the assets they are trading. All of these participants compete with one another in an attempt to achieve their own personal objectives in the most efficient way possible. Traded assets may also be driven by latent factors, and agents must dynamically incorporate data into their trading decisions.

In this paper, we propose a game theoretic model in which a large population of heterogeneous agents all trade the same asset. This model considers heterogeneity not only from the point of view of an individual's trading objectives and risk appetite, but also from the point of view of each agent's beliefs regarding the performance of the asset they are trading. We pay particular attention to the information each agent is privy to, in an attempt to render the framework as realistic as possible, while maintaining analytical tractability.
We study the equilibrium of these markets by using the theory of mean-field games (MFGs), which studies the system as the number of participating agents becomes arbitrarily large. The general theory of mean-field games already has a large body of research associated with it. The original works stem from \cite{huang2006large}, \cite{HuangCaines_TAC07}, and \cite{lasry2007mean}. {Among the many extensions and generalizations which explore the broad theory of MFGs as well as their applications, we highlight the following works:~\cite{huang2010large} and \cite{nourian2013e} who investigate MFGs with combinations major and minor agents, \cite{carmona2013probabilistic} who develop a probabilistic analysis of MFGs, as well as the works of~\cite{cirant2015multi,bensoussan2018mean} who introduce MFGs with heterogeneous populations of agents.}
This theory has seen applications in various financial contexts, such as \cite{gueant2011mean} who explores various applications of MFGs in economics, \cite{carmona2013mean} and \cite{huang2017robust} who study systemic risk, \cite{jaimungal2015mean} who studies algorithmic trading in the presence of a major agent and a population of minor agents,  \cite{cardaliaguet2016mean} who investigates optimal execution, and \cite{firoozi2015varepsilon}, \cite{firoozi2016mean} who look at MFGs with partial information on states and apply it to algorithmic trading.

{
Other works that study differing beliefs of market participants include \cite{bayraktar2017mini}, who study a system where agents' believe the asset price is an arithmetic Brownian motion with a latent (constant) drift, and agents disagree on the prior distribution of this latent drift, as well as on the temporary and permanent impact trading has on prices. The authors do not seek an equilibrium, but rather look at how the differences in belief may cause mini-flash crashes. \cite{bouchard2018equilibrium} study a model where agents, with differing risk aversion who receive random endowments, trade assets who's drifts are determined in equilibrium. Under certain assumptions, the equilibria results in asset prices having a permanent price impact component due to the existence of trading costs (temporary price impact). \cite{choi2018smart} study how traders who penalize deviations from a target strategy, and have their own private information, form an equilibria.
}

In contrast to other work on MFGs, as well as its specific application to algorithmic trading, here, motivated by \cite{casgrain_jaimungal_2016}, we include latent states so that agents do not have full information about the system dynamics.  Furthermore, motivated by \cite{firoozi2016mean} and \cite{casgrain2018algorithmic}, who study a stochastic game with latent factors where agents have the same model beliefs, here, we study how varying beliefs among the agents affect the optimal trading behaviour. In our model, we express the belief of agents as a probability measure on the dynamics of the asset price process and of any latent processes that may be driving them. As far as the authors are aware, this is the first time that MFG with varying beliefs have been treated in the literature. This generalization is quite non-trivial, nonetheless, we succeed in characterizing the model equilibrium as the solution to a non-standard forward-backward stochastic differential equation (FBSDE) defined across the collection of belief measures. We are able to present a closed form representation for the solution of the MFG and it incorporates all of the differing market's beliefs into the decisions of the individual agents.

Our key result, is the optimal mean-field trading rate $\bnubarst_t$ for the collection of sub-populations (within which agents have the same belief) can be written as
\[
		\bnubarst_t = \bg_{1,t} + \bg_{2,t} \, \bqbar_t^{\bnubarst},
\]
where $\bg_{2,t}$ is a deterministic matrix-valued process, $\bg_{1,t}$ is stochastic and encodes the various beliefs of the agents and their expectations of the future dynamics of the asset price, and $\bqbar_t^{\bnubarst}$ is the corresponding mean-field inventories. Moreover, the individual agents' trading rates within a subpopulation-$k$ can be written as
\[
\nujst_t = \nubarkst_t
+ \tfrac{1}{2 a_k} \, h_{2,t}^k
\,( \qj{\nujst}_t - \qbark{\nubarkst}_t )
			 \;,
\]
where $h_{2,t}^k$ is a deterministic function of time. Hence, agents speed up or slow down relative to the mean-field trading rate depending on whether their current inventory is above or below their subpopulation's mean-field inventory. The model setup does not penalize deviations from the mean-field, yet agents tend to revert to the mean-field -- this is in constrast to many MFG formulations where the deviation from the mean-field is explicitly penalized.

%
%

We structure the remainder of the paper as follows. Section~\ref{sec:The-Model-Description} introduces the market model and the stochastic game that agents participate in. Section~\ref{sec:Solving-The-MFG} begins by introducing the MFG limit of the stochastic game and then proves the collection of optimal strategies in the MFG may be represented as the solution to a system of coupled FBSDEs. Next, we solve the system of FBSDEs, and find the mean-field and each individual agents' strategy. Section~\ref{sec:Example-Model-Subsection} provides a specific example of a model where the assumptions in the key results are satisfied. In Section~\ref{sec:epsilon-Nash} we prove the solution to the MFG satisfies the $\epsilon$-Nash equilibrium property in the finite population game. Lastly, Section~\ref{sec:Numerical-Experiments} provides a least-square Monte Carlo approach to computing certain expecations, as well as simulated examples of a market model with agents having differing beliefs.

\section{The Model} \label{sec:The-Model-Description}

In this section, we provide the  market model and the participating agents' performance criteria. Our model closely resembles the model for the stochastic game in \cite{casgrain2018algorithmic}. The stochastic game  here aims to characterize a population of agents with several sources of heterogeneity. As in \cite{casgrain2018algorithmic}, here, agents have varying trading objectives. In addition, however, agents are also characterized by their beliefs regarding the model driving the asset price process. In the remainder of this section, we present the trading mechanics which each of the agents use to interact with the market, as well as the objectives each of the agents seek to achieve with their actions.

\subsection{The Population of Agents}

The market consists of a population of $N>0$ rational heterogeneous agents trading a single asset. Agents are indexed with an integer $j\in\mfN:=\{1,\dots,N\}$. The total population of agents is divided into $K\in\{1,\dots,N\}$ disjoint sub-populations, which are indexed by $k\in\mfK := \{1,\dots, K\}$. $K$ is assumed to be constant and independent of $N$. All agents within a fixed sub-population have homogeneous beliefs and performance criteria. The set
\begin{equation}
	\K_k^{\N} := \{ j \in \mfN \;\colon \text{ $j$ is in sub-population $k$}  \},
\qquad \forall k\in\mfK
	\;,
\end{equation}
denotes the set of agents within sub-population $k$, and the superscript $\N$ indicates the explicit dependence on the total number of agents. We also define $N_k^{\N}:=|\K_k^{\N}|$ to be the total number of agents within sub-population $k$. We further assume the number of agents contained in each of the sub-populations remains stable as we take the population limit to infinity. More specifically, we require that the proportion of agents contained within population $k$ satisfies
\begin{equation}
	\lim_{N\rightarrow\infty} p_k^{\N} = p_k \in (0,1)
\qquad\text{where}\qquad p_k^{\N}=\frac{N_k^{\N}}{N}
	\;.
\end{equation}

\subsection{The Agent's State Processes} \label{sec:The-State-Processes}

We work on the filtered probability space $(\Omega,\mathfrak{G}=\{\G_t\}_{\tT},\P)$ completed by the null sets of $\P$ and where $T\in(0,\infty)$ is some fixed time horizon. All of processes defined in the remainder of this section are $\G$-adapted, unless otherwise specified, and  the notation $\EP[\cdot]$ represents expectation with respect to the measure $\P$.

All agents have the ability to buy and sell the asset over the fixed trading period $[0,T]$, after which all trading activity comes to a halt. Each agent $j\in\mfN$ controls the amount they wish to purchase or sell at a continuous rate denoted $\nuj = (\nuj_t)_{\tT}$, where $\nuj_t>0$ ($\nuj_t<0$) indicates the rate of buy (sell) orders the agent sends to the market. {At the start of the trading period, each agent holds a random amount $\mfQ_0^j$ of the asset. This may be interpreted as each agent having private information about their own holdings, whereas other market participant know only its distribution.}
Agents keep track of their holdings in the traded asset with the inventory process $\qj{\nuj} = (\qj{\nuj}_t)_{\tT}$, where the superscript indicates the explicit dependence on the agent's controlled rate of trade. The relationship between agent-$j$'s  trading rate and their inventory process is
\begin{equation}
	\qj{\nuj}_t = \mfQ_0^j + \int_0^t \nuj_t \, dt\;,
\end{equation}
and may be interpreted as each agent buying or selling an amount $\epsilon \,\nuj_t$ in each small time interval $[t,t+\epsilon)$.
\begin{assumption}
We make the technical assumption that the initial inventory holdings of all agents have a bounded variance, so that $\exists \;0<\mathfrak{C}<\infty$ for which $\EP (\mfQ_0^j)^2 < \mathfrak{C}, \forall j \in \mfN$.
Moreover, we assume that the mean of the starting inventory levels are the same within a given sub-population, so that $\EP[\mfQ_0^j] = \overline{m}_k$ for each $j\in\K_k^{\N}$.
\end{assumption}

Buying and selling actions of agents impact the price of the traded asset in a manner to be specified below. As well, agents believe the asset midprice follows (potentially) different models. We incorporate differing beliefs into our model by assigning a probability measure $\Pk$ to each sub-population $k\in\mfK$.
The various measures correspond to the model that agents in a particular sub-population believes to represent the true dynamics of the asset price.

We define the asset price process $S^{\bnu}=(S^{\bnu}_t)_{\tT}$, where the superscript $\bnu=(\nu^j)_{j\in\mfN}$ indicates the dependence of the price on the actions of all agents in the market.  It is useful to define the average trading rate $\nubarkN=(\nubarkN_t)_{\tT}$  of all agents within sub-population $k$ as
\begin{equation}
	\nubarkN_t = \frac{1}{N_k^{\N}} \sum_{j\in\K_{k}^{\N}} \nuj_t\,.
\end{equation}
Each agent in sub-population $k$ then believes the asset price process follows the dynamics
\begin{equation} \label{eq:Midprice-Dynamics}
	S^{\bnu}_t = S_{0} +
	\int_0^t \left\{
	A_u^k +
	\sum_{k^\prime\in\mfK} \lambda_{k,k^\prime} \, p_{k^\prime}^{\N}\, \nubar^{k^\prime,\N}_u
	\right\} \, du + M_t^k
	\;,
\end{equation}
where for each $k\in\mfK$, $A^k = (A_t^k)_{\tT}$ is a $\G$-predictable process, $M^k = (M_t^k)_{[0,T]}$ is a $\G$-adapted $\Pk$-martingale, and $\lambda_{k,k^\prime} > 0$ $\forall k,k^\prime \in \K$ are constants. We also assume  the initial inventory holdings of each agent ${\mfQ_0^j}_{j\in\N}$ are all independent of both $\{A^k\}_{k\in\mfK}$ and $\{M^k\}_{k\in\mfK}$ in each measure $\Pk$.

The measure $\Pk$ effectively specifies the sub-population-$k$'s asset price model through the processes $A^k$ and $M^k$, as well as the scale of the market impact of each sub-population, through set of constants $\{\lambda_{k,k^\prime}\}_{k^\prime\in\K}$.
\begin{assumption}
We make the technical assumptions that $A^k\in\HTk$ and $M^k\in\LTk$, where
\begin{subequations}
\begin{align}
	\HTk &= \left\{ f_t:\Omega\times[0,T] \rightarrow \R \colon \Ek \int_0^T \lVert f_t \rVert^2 \, dt < \infty \right\}
\qquad
	 \text{and}
	\\
	\LTk &= \left\{ f_t:\Omega\times[0,T] \rightarrow \R \colon \Ek \lVert f_t \rVert^2 < \infty \;, \forall \tT \right\}
	\;,
\end{align}
\label{eqn:Integrability}%
\end{subequations}%
for each $k\in\mfK$ and where $\lVert \cdot \rVert$ represents the Euclidean norm.
\end{assumption}
\begin{assumption}
	We assume that $\Pk\sim\P$ for all $k\in\mfK$ and the law $\mfQ_0^j$ under each measure $\Pk$ is the same as that under the measure $\P$.
\end{assumption}

\begin{assumption}
	We assume that for each $k\in\mfK$, $A^k$ and $M^k$ are uncontrolled -- i.e., are unaffected by the agents' actions.
\end{assumption}

Our asset price process model does not require explicitly specifying $A^k$ and/or $M^k$ in advance. Rather, we only require the integrability conditions in \eqref{eqn:Integrability}, hence there is great flexibility in the class of models our approach accommodates. For example, $A^k$ and $M^k$ can be discontinuous or non-Markov, as long as they satisfy the appropriate integrability conditions. The key assumption we make is that price impact from the order-flow of all agents' trading is linear.

\begin{remark}
{Agents do not change their assigned belief measure, even after observing data from $(S^{\bnu}_t)_{\tT}$, however, their prior assumptions are updated to posterior estimates as time flows.}
\end{remark}


Each agent tracks their total accumulated cash process $X^{j,\nuj}_t = ( X_t^{j,\nuj} )_{\tT}$ throughout the trading period. When buying and selling the asset, each agent pays an instantaneous cost that is linearly proportional to amount of shares transacted. This cost is expressed through the controlled dynamics of the cash process. For an agent $j\in\K_k^{\N}$, their corresponding cash process is
\begin{equation}
	X^{j,\nuj}_t = X^{j}_0 - \int_0^t  \left( S_t^{\nubarN} + a_k \,\nuj_u  \right) \,\nuj_u \, du
	\;,
\end{equation}
where $a_k>0$ is a parameter that is unique to a sub-population $k$ and sets the scale of the instantaneous cost. {The penalty $-a_k\int_0^t(\nu_u^j)^2\,du$ may be interpreted as a cost of trading too quickly, or a tractable proxy for the cost of crossing the bid-ask spread, as in~\cite{bouchard2018equilibrium}. It is also straightforward to include the influence of other agents in this cost, i.e., replace $S_t^{\nubarN} + a_k \,\nuj_u$  by $S_t^{\nubarN} + \sum_{k\in\mfK} a_{k} \, p_{k}^{\N}\, \nubar^{k,\N}_u$.}

\subsection{Information Restriction} \label{sec:Info-Restriction}

In this market model, agents have restricted information over the course of the trading period. More specifically, agents have access only to the information generated by the paths of the asset price process $S^{\bnu}$, their own inventory process $q^{j,\nu^j}$, and the average order flow of each sub-population, $\bnubarkN = \left(\nubarkN\right)_{k\in\mfK}$.
We express this information restriction in our model by restricting the sigma-algebra to which an agent's strategy may be adapted. For each $j\in\mfN$, we only allow agent-$j$ to choose strategies contained within the set of asmissible strategies,
\begin{equation} \label{eq:Admissible-Set-Def}
	\A^j := \left\{ \omega \in \HT ,\; \omega \text{ is $\F^j$-predictable} \right\}
	\;,
\end{equation}
where we define $\HT = \bigcap_{k\in\mfK} \HTk$, and
\begin{equation}
	\F_t^j = \sigma\left( (S^{\bnu}_u, \bnubarkN_u )_{u\in[0,t)} \right)
	\vee
	\sigma\left( \mfQ_0^j \right)
	\;,
\end{equation}
which is the sigma-algebra generated by the paths of the asset price process, the total order-flow proces,, and the starting inventory level for agent $j$.
In definition~\eqref{eq:Admissible-Set-Def}, we deliberately restrict ourselves to processes in $\HT$, to guarantee that $S_t^{\bnu} \in \HTk$ for all $k\in\mfK$.

\subsection{The Agent's Optimization Problem} \label{sec:Optimization-Problem}

Each agent chooses their  trading strategy to maximize an objective functional that measures their performance over the course of the trading period $[0,T]$. For each $j\in\mfN$ let $\A^{-j}:=\bigtimes_{i\in\mfN , i\neq j} \A^i$. Each agent-$j$ within a sub-population $k\in\mfK$, chooses a control $\nu_j\in\A^j$ to maximize a functional $H_j:\A^j \times \A^{-j} \rightarrow \R$ defined as follows
\begin{equation} \label{eq:Objective-Function-Definition}
	H_j(\nuj,\nu^{-j})
	=
	\Ek \left[
	X_T^{\nuj} +
	\qj{\nuj}_T \left( S_T^{\nubarN} - \Psi_k \qj{\nuj}_T \right)
	- \phi_k \int_0^T ( \qj{\nuj}_u )^2 \, du
	\right],
\end{equation}
where $\Psi_k>0$ and $\phi_k \geq 0$ are parameters that vary across, but are constant within, sub-populations. In definition~\eqref{eq:Objective-Function-Definition}, we use the notation $\nu^{-j} := \left(\nu^1 , \dots , \nu^{j-1},\nu^{j+1},\dots,\nu^N\right)$ to indicate the dependence of the objective functional on the actions of all other agents in the population.

The objective functional corresponds to the agent trying to maximize a weighted average of three separate quantities. The first term $X_T^{\nuj}$ corresponds to the total amount of cash the agent has accumulated up until time $T$. The second term, $\qj{\nuj}_T \left( S_T^{\nubarN} - \Psi_k \qj{\nuj}_T \right)$ corresponds to the cost of liquidating all of the agent's leftover inventory at time $T$, minus a liquidation penalty controlled by the parameter $\Psi_k$. The last term, $- \phi_k \int_0^T ( \qj{\nuj}_u )^2 \, du$ is a running risk-aversion penalty that is controlled by the parameter $\phi_k$, which incentivizes the agent to keep their market exposure low during the trading period. {As demonstrated in~\cite{cartea2017algorithmic}, this term may also be interpreted as stemming from an agent's model uncertainty with respect to a continuum of measures, absolutely continuous with respect to the reference measure, and penalize those candidate measures with relative entropy.}

Each agent within sub-population $k$ has an objective functional that is computed by taking expectations under the measure $\Pk$. Hence, agents incorporate their own beliefs on the asset price dynamics. Furthermore, each functional $H_j$ depends on the actions of all other players ($\nu^{-j}$) through the dynamics of the asset price $S_t^{\bnu}$, which implicitly appear in the definition~\eqref{eq:Objective-Function-Definition}.

By expanding the dynamics of each of the state processes present in \eqref{eq:Objective-Function-Definition}, and by using integration by parts, we may re-write the agent's objective functional as
\begin{equation} \label{eq:Objective-Function-Alternative-Representation}
	H_j(\nuj,\nu^{-j})
	=
	C_0^j +
	\Ek \left[
	\int_0^T
	\qj{\nuj}_t dS_t^{\nubarN}
	-
	\begin{pmatrix}
		\nuj_t \\ \qj{\nuj}_t
	\end{pmatrix}^\T
	\begin{pmatrix}
	a_k & \Psi_k \\ \Psi_k & \phi_k
	\end{pmatrix}
	\begin{pmatrix}
		\nuj_t \\ \qj{\nuj}_t
	\end{pmatrix}
	\, dt
	\right],
\end{equation}
where $C_0^j$ is a term that is constant with respect to $\nuj$ and $\nu^{-j}$. Each agent's behaviour is characterized entirely by the objective functional they are trying to maximize. From \eqref{eq:Objective-Function-Alternative-Representation}, it is clear that the objective functional is parametric so that the agent's preferences can be entirely described by the tuple $\left( a_k , \phi_k , \Psi_k , \Pk \right)$ and their starting inventory $\mfQ_0^j$.

The market model  defined above forms a stochastic game in which all participating agents are competing to maximize each of their own objectives. We wish to find and study this market at its Nash equilibrium. This equilibrium can be described more formally as the collection of admissible strategies $\{\nuj \in \A^j \colon j\in\mfN\}$ which satisfies the condition
\begin{equation} \label{eq:Finite-Nash-Equilibrium-Def}
	\nujst = \argmax_{\omega\in \A^j} H_j\left(\omega,\nu^{-j}\right)
	\,, \hspace{2em} \forall j \in \mfN
	\;.
\end{equation}
Obtaining this collection of strategies for the stochastic game with a finite number of players proves to be a difficult task.
{As we make no assumptions on $\{\nuj\}$ beyond the measurability and integrability assumptions required in the definition of $\A^j$, and in particular do not use a feed-back from, a Nash equilibrium satisfying~\eqref{eq:Finite-Nash-Equilibrium-Def} will be an open-loop equilibrium in general.} One of the main obstacles in finding a solution to this problem is that each agent's strategy is adapted to different filtration $\F^j$. Furthermore, each of the objective functionals defined in equation~\eqref{eq:Objective-Function-Alternative-Representation} are expressed one of $K$ different measures from the collection of measures $\{\Pk\}_{k\in\mfK}$, each representing the beliefs of a particular individual. These two features make the finite-population stochastic game difficult to solve directly. It is, however, possible to solve the stochastic game in the infinite population limit, and use the result as an approximation for the finite population game.

\section{Solving the Mean-Field Stochastic Game} \label{sec:Solving-The-MFG}

As the stochastic game presented in Section~\ref{sec:Optimization-Problem} presents obstacles when aiming to solve it directly, we now take a different avenue. In this section, we study the stochastic game as the population limit tends towards infinity. The resulting limit is that of a stochastic Mean Field Game (MFG) that we can solve. Although we do not explicitly solve the finite player game presented in Section~\eqref{sec:The-Model-Description}, by establishing an $\epsilon$-Nash equilibrium property in Section~\ref{sec:epsilon-Nash}, we show that the equilibrium solution obtained for the MFG provides an approximation to the finite population game, provided that the population size is large enough.

This section begins by taking the population limit as $N\rightarrow\infty$, to obtain new objective functionals for the agents resulting in a stochastic MFG. Next, using convex analysis methods, we characterize the Nash-equilibrium as the solution to a coupled system of FBSDEs. We then conclude by presenting a solution to this FBSDE problem, and thus an exact representation of each agent's optimal control at the Nash-equilibrium.

\subsection{The Limiting Mean-Field Game}
\label{sec:Limiting-MFG}

Agent-$j$'s objective functional \eqref{eq:Objective-Function-Definition} only depends on the population size $N$ through the dynamics of the mid-price process $S_t^{\nubarN}$, which is given by the dynamics in equation~\eqref{eq:Midprice-Dynamics}.
\begin{assumption}
To proceed, we assume that the limiting trading rate exists, in particular, there exist processes $\nubark = (\nubark_t)_{\tT}$ for $k\in\mfK$ such that $\nubark\in\HT$ and
\begin{equation} \label{eq:Limiting-Mean-Field-Assumption}
	\lim_{N\rightarrow\infty} \nubarkN_t = \nubark_t
	\,,\qquad
	\P \times \mu
	\text{ a.e.,}
\end{equation}
where $\mu$ is the Lebesgue measure on the Borel sigma-algebra $\mcB_{[0,T]}$, and where $\P\times\mu$ is the canonical product measure of $\P$ and $\mu$.
\end{assumption}
As each individual $\nuj$ is $\F^j$-predictable, $\nubark$ must be $\left(\bigvee_{j\in\mfN} \F^j\right)$-predictable.
Moreover, by our assumption that $\Pk\sim\P$ for each $k\in\mfK$, the limit~\eqref{eq:Limiting-Mean-Field-Assumption} also holds $\Pk\times \mu$ almost everywhere. From now on, we refer to each of the processes $\nubark$ as the mean-field trading rate for sub-population-$k$.

Using the assumption that $p_k^{\N} \rightarrow p_k$ for all $k\in\mfK$ along with \eqref{eq:Limiting-Mean-Field-Assumption}, we find that in the infinite population limit, from the perspective of agent-$j$ from sub-population $k$, the dynamics of the asset price process is
\begin{equation} \label{eq:Limiting-Midprice-Dynamics}
	S^{\nubar}_t = S_{0} +
	\int_0^t \left\{
	A_u^k +
	\sum_{k^\prime\in\mfK} \lambda_{k,k^\prime} \, p_{k^\prime}\, \nubar^{k^\prime}_u
	\right\} \, du + M_t^k
	\;.
\end{equation}
{In this limit, a single individual's impact on the price becomes negligible, thus the resulting mean-field trading rate $\nubark$ is unaffected by a single agent's trading rate $\nuj$. Therefore, in the limit, each agent's objective $H_j$ no longer depends on the whole collection of trading rates $\nu^{-j}$, but instead only depends on the collection of mean-field processes $\{\nubark\}_{k\in\mfK}$, which considerably simplifies the dependence structure within the game. This can be interpreted as agents becoming `price takers' in the mean-field limit resulting from the aggregate of agents' infinitesimal impact.}

By using the objective functional representation in~\eqref{eq:Objective-Function-Alternative-Representation}, expanding $dS_t^{\nubar}$ from~\eqref{eq:Limiting-Midprice-Dynamics}, and noticing that the martingale components vanish under expectation, we may write the agents objective functional in the infinite population limit as
\begin{equation} \label{eq:Objective-Function-Limit}
	\Hbar_j^{\bnubar}(\nuj)
	=
	C_0^j +
	\Ek \left[
	\int_0^T
	\left\{
	\qj{\nuj}_t
	\left( A_t^k + \blambda_k^\T \, \bnubar_t \right)
	-
	\begin{pmatrix}
		\nuj_t \\ \qj{\nuj}_t
	\end{pmatrix}^\T
	\begin{pmatrix}
	a_k & \Psi_k \\ \Psi_k & \phi_k
	\end{pmatrix}
	\begin{pmatrix}
		\nuj_t \\ \qj{\nuj}_t
	\end{pmatrix}
	\right\}
	\, dt
	\right],
\end{equation}
where for each $k\in\mfK$ we define $\blambda_k\in\R^K$ as $\blambda_k=(\lambda_{k,k^\prime} \, p_{k^\prime})_{k^\prime \in \K}$ and where we define $\bnubar_t\in \R^K$ as $\bnubar_t = (\nubark)_{k\in\mfK}$. In the expression for $\Hbar^{\bnubar}_j$ in~\eqref{eq:Objective-Function-Limit} we suppress the argument $\nu^{-j}$ as, in this infinite population limit, their effect is felt through the mean-fields for each subpopulation. We use the superscript in the notation for $\Hbar_j^{\bnubar}$ to indicate the dependence on the set of mean-fields.

Our new objective is to obtain the Nash-equilibrium in this newly defined mean-field game. The Nash equilibrium for the MFG consists of finding the infinite collection of controls $\{\nuj\}_{j=1}^\infty$ that satisfies the optimality condition
\begin{equation} \label{eq:MFG-Optimization-Problem}
	\nujst = \argmax_{\omega\in\A^j}\Hbar_j^{\bnubarst} (\omega)\,,
\end{equation}
as well as the consistency condition
\begin{equation}
\label{eqn:MFG-consistency-requirement}
	\nubarkst = \lim_{N \nearrow \infty} \frac{1}{N_k^{\N}} \!\! \sum_{j\in\K_k^{\N}} \nujst_t\,,
\end{equation}
for all $k\in\mfK$.

In the limit, the explicit dependence of an agent's actions in another agent's  objective functional is replaced with an implicit dependence through the consistency condition.

\subsection{The Agent's Optimality Condition}

To solve the optimization problem described in Section~\ref{sec:Limiting-MFG}, we must determine what strategy maximizes the rhs of equation~\eqref{eq:MFG-Optimization-Problem} for all agents. This is achieved by using tools from infinite dimensional convex-analysis or variational calculus along the lines of \cite{bank2017hedging} (who investigate a single agent tracking problem that does not incorporate price information) and \cite{casgrain2018algorithmic} (who look at a multi-agent setting with price imformation, but where all agents use the same model). First, we demonstrate that each function $\Hbar^{\bnubar}_j$ is a strictly concave functional of $\nuj$. Next, as $\Hbar^{\bnubar}_j$ is a functional with an infinite-dimensional argument, we show that each functional $\Hbar^{\bnubar}_j$ is G\^ateaux differentiable within the space $\A^j$ and compute the G\^ateaux derivative explicitly. General results in convex optimization then state that if the derivative vanishes at a point within the space $\A^j$, it must be the point at which $\Hbar^{\bnubar}_j$ attains its supremum. The lemmas that follow give us the required properties for $\Hbar^{\bnubar}_j$.
\begin{lemma} \label{thm:Convex-Lemma}
	The functional $\Hbar^{\bnubar}_j$ defined in equation~\eqref{eq:Objective-Function-Limit} is strictly concave in $\A^j$ up to $\P\times\mu$ null sets.
\end{lemma}
\begin{proof}
See \ref{sec:Proof-Convex-Lemma}.
\end{proof}

\begin{lemma} \label{thm:Differentiable-Lemma}
	For an agent-$j$ in sub-population $k$, the functional $\Hbar^{\bnubar}_j$ defined in equation~\eqref{eq:Objective-Function-Limit} is everywhere G\^ateaux differentiable in $\A^j$. The G\^ateaux derivative at a point $\nu\in\A^j$ in a direction $\omega\in\A^j$ can be expressed as
\begin{equation} \label{eq:Gateaux-Derivative-Expression}
\begin{split}
	\left\langle \D \Hbar_j^{\bnubar}(\nu) , \omega \right\rangle
	= &
	\Ek \left[
	\int_0^T
	\omega_t \,
	\Bigg\{
	-2a_k\nu_t - 2\Psi_k \qj{\nu}_T \right.
\\
& \left.\qquad\qquad\qquad
+ \int_t^T \left(
	\,\Ek\left[ A^k_u + \blambda_k^\T\, \bnubar_u \,\lvert \, \F_u^j \right] - 2\phi_k \,\qj{\nu}_u \right) \, du
	\Bigg\}
	\, dt
	\right]
	\;.
\end{split}
\end{equation}
\end{lemma}
\begin{proof}
See \ref{sec:Proof-Differentiable-Lemma}.
\end{proof}

Therefore, since $\Hbar^{\bnubar}_j$ is concave, the supremum of $\Hbar^{\bnubar}_j$ is attained at a point $\nu\in\A^j$ if and only if the expression~\eqref{eq:Gateaux-Derivative-Expression} vanishes for all   $\omega\in\A^j$. Moreover, the strict concavity of $\Hbar^{\bnubar}_j$ guarantees that such a point is unique up to $\P\times\mu$ null sets.
Indeed, as the following theorem shows, the collection of points $\{\nuj\}_{j=1}^\infty$ that ensures~\eqref{eq:Gateaux-Derivative-Expression} vanishes for all $j\in\mfN$, and for all $\omega\in\A^j$, coincides with the solution of an infinite-dimensional system of FBSDE.
\begin{theorem} \label{thm:FBSDE-Optimality-Condition}
We have that
	\begin{equation}
		\nujst := \argmax_{\nu\in\A^j} \Hbar^{\bnubarst}_j(\nu)
	\end{equation}
for all $j\in\mfN$ if and only if for each agent-$j$ in sub-population $k$, $\nujst\in\HT$ and $\nujst$ is the unique strong solution to the FBSDE
\begin{equation} \label{eq:Optimal-FBSDE-System}
\left\{
\begin{split}
-d(2\,a_k\,\nujst_t)
&= \left(\, \Ek\!\left[\,
A_t^k + \blambda_k^\T\, \bnubarst_t \,\lvert\, \F_t^j
\,\right]
- 2\,\phi_k \,\qj{\nujst}_t \right) dt - d\mcM_t^j\,,
		\\
2\,a_k \, \nujst_T &= -2\,\Psi_k \,\qj{\nujst}_T\,,
\end{split}
\right.
\end{equation}
where $\mcM^j\in\HT$ is an $\F^j$-adapted $\Pk$-martingale and where
\begin{equation} \label{eq:Optimality-Consistency-Condition}
	\nubarkst_t = \displaystyle\lim_{N\rightarrow\infty} \tfrac{1}{N_k^{\N}} \sum_{j\in\K^{\N}_k} \nujst_t
	\;,
\end{equation}
for all $k\in\mfK$.
\end{theorem}
\begin{proof}See \ref{sec:Proof-FBSDE-Optimality-Condition}.
\end{proof}

Theorem~\ref{thm:FBSDE-Optimality-Condition} reduces the convex optimization problem \eqref{eq:Objective-Function-Limit}, \eqref{eq:MFG-Optimization-Problem}, and \eqref{eqn:MFG-consistency-requirement} into an infinite system of FBSDEs. The forward component comes from the latent drift processes $A^k$ and inventory processes $\qj{\nujst}$, while the backwards component comes from the trading rates $\nujst$. The coupling in this system appears through the mean-field processes $\bnubar$, which averages out all of the actions of other agents within the game. A few difficulties are immediately apparent in the FBSDE~\eqref{eq:Optimal-FBSDE-System}. Firstly, each individual FBSDE, corresponding to a particular agent's trading rate, is written in terms of a martingale that is specific to the agent's sub-population, and the measure under which the process is a martingale corresponds to the agent's belief about the drift process $A^k$. Secondly, the conditional expected value $\Ek\left[ \,A_t^k + \blambda_k^\T \,\bnubarst_t \,\lvert \,\F_t^j\right]$ appears in the driver of the FBSDE. This is a projection of the mean-fields onto the agent's filtration, and appears because the agent cannot directly observe the strategies of other individuals. This projection of the mean-fields adds another layer of difficulty.

Recall that a solution to the FBSDE~\eqref{eq:Optimal-FBSDE-System} for agent-$j$ consists of a pair of processes $(\nujst,\mcM^j)$ that satisfies the SDE and terminal condition in~\eqref{eq:Optimal-FBSDE-System} $\mathbb{P}\times\mu$ almost everywhere. For the requirements of Theorem~\ref{thm:FBSDE-Optimality-Condition} to be met, a solution must simultaneously meet the consistency condition~\eqref{eq:Optimality-Consistency-Condition} $\mathbb{P}\times\mu$ almost everywhere. If we can find a set of solutions, we can guarantee it is unique up to $\P\times \mu$ null sets due to the strict convexity of the objective functional and the `if and only if' nature of the statement.

\subsection{Solving the Optimality FBSDE}
\label{sec:Solving-The-Optimality-FBSDE}

In this section, we solve the FBSDE~\eqref{eq:Optimal-FBSDE-System}, and hence provide an exact form for the Nash-equilibrium for the infinite population mean-field game. The key to obtaining a solution lies in first postulating a structure for the solution of~\eqref{eq:Optimal-FBSDE-System}. This form then suggests a vector valued FBSDE that the mean-field processes $\nubark$ must satisfy, which are independent of any individual agent's strategy. The resulting non-standard FBSDE system, is defined across the set of $K$ measures $\{\Pk\}_{k\in\mfK}$ and introduces an obstacle in solving it directly. The key step in obtaining a solution lies in representing the FBSDE in terms of a single measure, and solving it there.

Due to the linear form of the FBSDE~\eqref{eq:Optimal-FBSDE-System}, it is natural to assume that the solution is affine. As such, for an agent-$j$ within a sub-population $k$, we seek for optimal controls of the form
\begin{equation} \label{eq:Trader-j-Ansatz-Structure}
	2 \,a_k \,\nujst_t = 2 \,a_k \,\nubarkst_t + h_{2,t}^k\, \left( \qj{\nujst}_t - \qbark{\nubarkst}_t \right),
\end{equation}
where $h^k_{2,t}:[0,T]\rightarrow \R$ is an unknown deterministic, continuously differentiable, function of time, and where we define the mean-field inventory process $\qbark{\nubarkst} = (\qbark{\nubarkst}_t)_{\tT}$ for sub-population $k$ as
\[
\qbark{\nubarkst}_t = \mbar_k + \int_0^t \nubarkst_u \, du\,.
\]
Plugging this ansatz into~\eqref{eq:Optimal-FBSDE-System} and simplifying, we find that
\begin{equation} \label{eq:Ansatz-Optimal-FBSDE-Plugged-In}
\begin{split}
	0 = \hphantom{+}&\,
	\Big\{ \partial_t h_{2,t}^k + \tfrac{1}{2a_k} (h_{2,t}^k)^2 - 2\,\phi_k \Big\}
	\left( \qj{\nujst} - \qbark{\nubark}_t \right) \, dt
	\\ +
	&\left\{
	d(2\, a_k \,\nubarkst_t) +
	\left( \Ek[\,A_t^k + \blambda_k^\T \,\bnubarst_t \,\lvert \F_t^j] - 2\,\phi_k \,\qbark{\nubarkst}_t \right) \,dt - d\mcM_t^j
	\right\}
	\;,
\end{split}
\end{equation} \label{eq:Ansatz-Optimal-Boundary-Plugged-In}
along with the boundary condition that
\begin{equation}
	0 =  \{ h_{2,T}^k + 2\,\Psi_k \} ( \qj{\nujst}_T - \qbark{\nubarkst}_T )
	+ \{ 2 a_k \nubarkst_T + 2\Psi_k \qbark{\nubarkst}_T \}
	\;,
\end{equation}
which must both hold $\Pk \times \mu$ almost everywhere. Therefore, to solve the FBSDE~\eqref{eq:Optimal-FBSDE-System}, it is sufficient for us to make the terms in the curly brackets of equation~\eqref{eq:Ansatz-Optimal-FBSDE-Plugged-In} and in the boundary condition~\eqref{eq:Ansatz-Optimal-Boundary-Plugged-In} vanish independently of one another. Collecting these equations, we obtain a first-order Riccati-type ODE for $h_{2,t}^k$,
\begin{equation} \label{eq:h2k-ODE-Equation}
\left\{
\begin{split}
		 \partial_t h_{2,t}^k +  \tfrac{1}{2a_k} (h_{2,t}^k)^2 - 2\phi_k =&\,0\,, \\
		h_{2,T}^k =& -2\Psi_k\,,
\end{split}
\right.
\end{equation}
as well as a linear FBSDE for the mean-field process $\nubark_t$
\begin{equation} \label{eq:Mean-Field-FBSDE-1}
\left\{
\begin{split}
-d(2\, a_k \,\nubarkst_t) =&
\left( \Ek[\,A_t^k + \blambda_k^\T \,\bnubarst_t\, \lvert \,\F_t^j] - 2\,\phi_k \,\qbark{\nubarkst}_t \right) \,dt - d\mcM_t^j\,,
		\\
2\,a_k \,\nubarkst_T =& -
2\,\Psi_k \,\qbark{\nubarkst}_T\,,
\end{split}
\right.
\end{equation}
where $\mcM^j = \left( \mcM^j_t \right)_{t\in [0,T]}\in\HT$ is an $\F^j$-adapted $\Pk$-martingale.

Let us point out here that the ansatz for $\nujst$ found in equation~\eqref{eq:Trader-j-Ansatz-Structure} satisfies the consistency condition as long as there exist solutions to the equations~\eqref{eq:h2k-ODE-Equation} and \eqref{eq:Mean-Field-FBSDE-1}. This can be most easily seen by taking the average of \eqref{eq:Trader-j-Ansatz-Structure} over $j\in\K_k^{N}$ and taking the limit as $N\rightarrow\infty$.

{The FBSDE~\eqref{eq:Mean-Field-FBSDE-1} suggests that the solution $\nubarkst$ should be an $\F^j$-adapted process. Equation~\eqref{eq:Mean-Field-FBSDE-1}, however, holds for any agent-$j'$ for which $j'\in\K_k$, therefore, $\nubarkst$ must be $\F^{j'}$-adapted for any $j'\in\K_k$. Consequently, each $\nubarkst$ must be adapted to the filtration generated by the intersection $\bigcap_{j'\in\K_k} \F_t^{j'}$. Computing this intersection, we find that $\bigcap_{j'\in\K_k} \F_t^{j'} = \bigcap_{j'\in\K_k} \sigma\left( (S_u,\bnubarst_u,q_u^{j,\nujst})_{u\in[0,t]}\right) \subseteq \sigma\left( (S_u,\bnubarst_u)_{u\in[0,t]}\right)$, which does not depend on the sub-population $k$. This is easy to see since (\textit{i}) each $q_0^j$ is not measurable with respect to $\sigma(q_0^i)$ for any $i\neq j$ and (\textit{ii}) for any $j\in\mfN$, $\nu^j$ is not measurable with respect to $\sigma(\bnubarst)$ by definition from~\eqref{eq:Optimality-Consistency-Condition}. Thus, for each $k\in\K$, we have that $\nubarkst$ is an $\F$-adapted process, where we define $\F_t := \bigwedge_{j\in\K_k} \F^j_t = \sigma\left( (S_u,\bnubarst_u)_{u\in[0,t]}\right)$.} As a consequence, we find that $\nubarkst$ should satisfy the FBSDE
\begin{equation}\label{eq:Mean-Field-FBSDE-2}
\left\{
\begin{split}
-d(2\, a_k \,\nubarkst_t) =&
\left( \Ek[\,A_t^k + \blambda_k^\T \,\bnubarst_t\, \lvert \,\F_t] - 2\,\phi_k \,\qbark{\nubarkst}_t \right) \,dt - d\mcMbar_t^k\,,
		\\
2\,a_k \,\nubarkst_T =& -
2\,\Psi_k \,\qbark{\nubarkst}_T\,,
\end{split}
\right.
\end{equation}
where $\mcMbar^k=(\mcMbar_t^k)_{\tT}$ is an $\F$-adapted, $\Pk$-martingale, and the expectation appearing in the drift is conditional on $\F_t$ not $\F^j_t$.

By stacking the FBSDEs~\eqref{eq:Mean-Field-FBSDE-2} over all values of $k\in\mfK$, we may obtain a vector-valued FBSDE for the process $\bnubarst$. To this end, define the column vector of filtered drift processes $\bAhat=(\bAhat_t)_{\tT}$ where $\bAhat_t = \left( \Ek[A_t^k \lvert \F_t] \right)_{k\in\mfK}$. Next, as $\bnubarst_t$ is $\F_t$-measurable, stacking the FBSDEs~\eqref{eq:Mean-Field-FBSDE-2} over all values of $k\in\mfK$, we have
\begin{equation} \label{eq:Mean-Field-FBSDE-Stacked}
\left\{
\begin{split}
-d(2 \,\ba \, \bnubarst_t) =&
		\left( \bAhat_t + \bLambda \, \bnubarst_t - 2\,\bphi \, \bqbar^{\bnubarst}_t \right) \,dt - d\bmcMbar_t\,,
		\\
2 \, \ba \, \bnubarst_T =& - 2 \, \bPsi \, \bqbar^{\bnubarst}_T\,,
\end{split}
\right.
\end{equation}
where $\ba,\bphi,\bPsi$ and $\bLambda$ are all real-valued $K \times K$ matrices defined as
\begin{align*}
	\ba &= \text{diag}\left( \{a_k\}_{k\in\mfK}\right) \,, \hspace{3em}
	\bphi = \text{diag}\left( \{\phi_k\}_{k\in\mfK}\right)\,,  \\ \vspace{5cm}
	\bPsi &= \text{diag}\left( \{\Psi_k\}_{k\in\mfK}\right) \,, \hspace{3em}
	\bLambda =
	\begin{pmatrix}
		\lambda_{1,1} \, p_1 & \dots & \lambda_{1,K} \, p_K \\
		\vdots &  & \vdots \\
		\lambda_{K,1} \, p_1 & \dots & \, \lambda_{K,K} p_K
	\end{pmatrix}
	\;,
\end{align*}
where $\bqbar^{\bnubarst}_t = \boldsymbol{\invmean}_0 + \int_0^t \bnubarst_u \, du$, and  $\bmcMbar = ({\mcMbar}^k)_{k\in\mfK}$ is a column vector of the $\F$-adapted processes, where as a reminder, $\mcMbar^k_t \in \HT$, $\forall k\in\mfK$ and the $k$-th element $\mcMbar^k$ is a $\Pk$-martingale.

From the linear structure of the FBSDE~\eqref{eq:Mean-Field-FBSDE-Stacked}, we can further simplify the problem by seeking for affine solutions of the form
\begin{equation} \label{eq:Ansatz-Mean-Field}
	2 \ba \bnubarst_t =
	\bg_{1,t} +
	\bg_{2,t} \, \bqbar_t^{\bnubarst}\,,
\end{equation}
where $\bg_{2,t}:[0,T] \rightarrow \R^{K \times K}$ is a deterministic and continuously differentiable function of time, and $\bg_{1} = (\bg_{1,t})_{\tT}\in\HT$ is an $\R^k$-valued stochastic process. Plugging the ansatz into~\eqref{eq:Mean-Field-FBSDE-Stacked}, and following through with the same logical steps as before, we find that the ansatz holds true so long as $\bg_{2}$ is the solution to the Ricatti-type matrix-ODE
\begin{equation} \label{eq:g2-ODE}
	\begin{cases}
		&\partial_t \bg_{2,t} = \left(\bLambda + \bg_{2,t} \right)
		\left( 2 \ba \right)^{-1} \bg_{2,t} - 2 \bphi\,,
		\\
		&\hspace{0.6em}\bg_{2,T} = - 2\bPsi\,,
	\end{cases}
\end{equation}
and when $\bg_{1,t}$ solves the BSDE,
\begin{equation} \label{eq:g1-BSDE}
	\begin{cases}
		&-d\bg_{1,t} = \left( \Ahat_t +
		\left( \bLambda + \bg_{2,t} \right) \left( 2 \ba \right)^{-1} \bg_{1,t}
		\right) \, dt
		- d\bmcMbar_t\,,
		\\
		&\hspace{1.1em} \bg_{1,T} = 0\,,
	\end{cases}
\end{equation}
where $\bmcMbar$ is the same vector of processes present in FBSDE~\eqref{eq:Mean-Field-FBSDE-Stacked}.

At this point, we have succeeded in reducing the search for a Nash-equilibrium to solving (i) two deterministic ordinary differential equations (ODEs)~\eqref{eq:h2k-ODE-Equation} and \eqref{eq:g2-ODE}, and (ii) a non-standard linear BSDE~\eqref{eq:g1-BSDE}. The ODEs are straightforward to solve, however, BSDE poses some further challenges.

{One of the primary obstacles in solving the BSDE~\eqref{eq:g1-BSDE} is that each component of $\bg_{1}$ incorporates a process that is a martingale under a different probability measure. Recall that the components of $\bmcMbar = \{\mcMbar^k\}_{k\in\mfK}$ are required to be martingales with respect to the $k$ different measures $\{\Pk\}_{k\in\mfK}$. Each measure is what agents within sub-population $k$ use to compute expectations, and agents within that sub-population assume the asset has drift $A^k$ in excess of the order-flow from all agents. The key step in solving the BSDE is to re-cast it in terms of martingales under a single probability measure. This introduces non-trivial drfit adjustments, however, we find that it is indeed possible to solve the modified BSDE explicitly.}

Consider the $k$th dimension of the BSDE~\eqref{eq:g1-BSDE}
\begin{equation}
\label{eqn:BSDEgk}
	-d g_{1,t}^k = \left( \Ahat_t^k + \bG_t^k \,\bg_{1,t}^k \right) \, dt
	- d{\mcMbar}^k_t
	\;,
\end{equation}
where ${\mcMbar}^k$ is a $\Pk$-martingale, and where $\bG_t^k$ is defined as the $k$-th row of the deterministic matrix-valued function $\bG_t = \left(\bLambda + \bg_{2,t}\right) \left( 2 \ba \right)^{-1}$. The solution of BSDE~\eqref{eqn:BSDEgk} can be expressed implicitly as follows
\begin{equation} \label{eq:g1k-Pk-Implicit-Solution}
	g_{1,t}^k =
	\E^\Pk \left[\left.
\,
	\int_t^T
	\left\{
	\Ahat_u^k + \bG_u^k \,\bg_{1,u}
	\right\}
	\, du
	\, \right\lvert \,\F_t\,
	\right].
\end{equation}

Next, we aim to represent~\eqref{eq:g1k-Pk-Implicit-Solution} in terms of expectation under another measure $\Q$ such that $\Q\sim\P^k$ for all $k$. By the assumption that $\P^k \sim \P^{k^\prime}$ for all $k,k^{\prime}\in\mfK$, there always exists such a measure. For example, $\Q=\Pk$ for some $k$. Given this measure, define the $\F$-adapted Radon-Nikodym derivative processes
\begin{equation}
	Z_t^{\Q,k} = \left.\frac{d\Pk}{d\Q}\right\lvert_{\F_t}:=\E\left[\left. \frac{d\Pk}{d\Q}\,\right|\,\F_t\right], \quad \forall k \in \mfK\,.
\end{equation}
Using this process, we find that we may write equation~\eqref{eq:g1k-Pk-Implicit-Solution} as an expected value under the $\Q$ measure as,
\begin{equation} \label{eq:g1k-Q-Implicit-Solution}
	Z_t^{\Q,k} g_{1,t}^k =
	\E^\Q \left[\left.
	\int_t^T
	\left\{
	 Z_u^{\Q,k} \,\Ahat_u^k +  Z_u^{\Q,k} \, \bG_u^k\, \bg_{1,u}
	\right\}
	\, du
	\,\right\lvert \,\F_t
	\right]
	\;.
\end{equation}

Defining the diagonal $\R^{K \times K}$ valued process $\bZ^{\Q}=(\bZ^{\Q}_t)_{\tT}$, where $\bZ^{\Q}_t = \text{diag}(Z_t^{\Q,k})_{k\in\mfK}$, allows us to write a linear BSDE for $\bZ_t^{\Q} \bg_{1,t} = \left(Z_t^{\Q,k} g_{1,t}^k\right)_{k\in\mfK}$ using a single measure $\Q$. More specifically, from~\eqref{eq:g1k-Q-Implicit-Solution}, we have that
\begin{equation} \label{eq:g1-BSDE-Q-measure}
	-d\left( \bZ^{\Q}_t \bg_{1,t} \right)
	= \left( \bZ^{\Q}_t \bAhat_t +
	\bZ^{\Q}_t \bG_t \, \bg_{1,t}
	\right) \,  dt
	- d\bmcMt_t
	\;,
\end{equation}
where $\bmcMt = (\bmcMt_t)_{\tT}$ is an $\R^K$-valued $\Q$-martingale.
The BSDE~\eqref{eq:g1-BSDE-Q-measure} is linear and its solution can be expressed in closed form. The following theorem provides a {representation for the solution of $\bg_{1}$ as well as $\{h_{2}^k\}_{k\in\mfK}$, and $\bg_{2}$.

\begin{theorem}[Solutions to the Mean-Field BSDEs]
\label{thm:Solution-g1-g2-h2}
~\\
\begin{enumerate}[label=\emph{\Roman*})]

\item Let $\Q$ be any probability measure such that $\Q \sim \P$. Then the BSDE~\eqref{eq:g1-BSDE} admits a closed form solution,
\begin{equation} \label{eq:proposition-solution-g1}
	\bg_{1,t} = \E^{\Q}\left[\int_t^T
			\left. (\bmcE_t^{\Q})^{-1} \, \bmcE_u^{\Q}
			 \,\bAhat_u
			\, du \, \right \lvert \F_t \, \right]
	\;,
\end{equation}
where $\bmcE_t$ is the solution to the forward matrix-valued SDE
\begin{equation} \label{eq:proposition-mcE}
	d\bmcE_t^{\Q} = \bmcE_t^{\Q} \left( \,
	\bG_t \, dt + (\bZ_t^{\Q})^{-1} d\bZ_t^{\Q} \,
	\right), \qquad \bmcE_0^{\Q} = \bZ_0^{\Q},
\end{equation}
where the deterministic matrix valued function $\bG_t := \left(\bLambda + \bg_{2,t}\right) \left( 2 \ba \right)^{-1}$ and
\begin{equation} \label{eq:Def-ZtQ}
	\bZ_t^{\Q} = \text{diag}\left(
\left.
\frac{d\Pk}{d\Q} \right \lvert_{\F_t}
\right)_{k\in\mfK}.
\end{equation}

\item There exists a unique solution $\bg_{2,t}$ to the matrix valued ODE~\eqref{eq:g2-ODE} that is bounded over the interval $[0,T]$.

Moreover, let $\bm{Y}_t:[0,T]\rightarrow\R^{2K\times K}$ be defined as
	\begin{equation}
		\bm{Y}_t =
		e^{(T-t)\bm{B}}
		\left( \bm{I}^{(K \times K)} ,\; -2\,\bPsi \right)^\T,
	\end{equation}
	where $\bm{B} \in \R^{2K \times 2K}$ is the block matrix
	\begin{equation}
		\bm{B} =
		\begin{pmatrix}
			\bm{0}^{(K\times K)} & -(2\ba)^{-1} \\
			-2\bphi & \bLambda (2\ba)^{-1}
		\end{pmatrix},
	\end{equation}
then, using the matrix partition $\bm{Y}_t = \left( \bm{Y}_{1,t} , \bm{Y}_{2,t} \right)^\T$, where $\bm{Y}_{1,t},\bm{Y}_{2,t} \in \R^{K\times K}$, the function $\bm{g}_{2,t}$ may be expressed as
\begin{equation}\label{eqn:proposition-solution-g2}
	\bg_{2,t} = \bm{Y}_{2,t} \, \bm{Y}_{1,t}^{-1}
	\;.
\end{equation}

\item The ODE~\eqref{eq:h2k-ODE-Equation} admits the unique solution
\begin{equation}
\label{eqn:proposition-solution-h2k}
		h_{2,t}^k = -2 \xi_k \left(
		\frac{
		\Psi_k \cosh\left( -\gamma_k (T-t)\ \right)
		- \xi_k \sinh\left( -\gamma_k (T-t)\ \right) }{
		 \xi_k \cosh\left( -\gamma_k (T-t)\ \right)
		- \Psi_k \sinh\left( -\gamma_k (T-t)\ \right)
		}
		\right), \qquad \forall k\in\mfK,
	\end{equation}
	where the constants $\gamma_k = \sqrt{\phi_k/a_k}$ and $\xi_k=\sqrt{\phi_k a_k}$. Moreover, $h_{2,t}^k \leq 0$ for all $\tT$.

\end{enumerate}

\end{theorem}
\begin{proof}
See \ref{sec:Proof-Solution-g1-g2-h2}.
\end{proof}

This theorem shows that $\bg_1$ may be expressed in terms of any measure $\Q\sim\P$, which includes any of the $\{\Pk\}_{k\in\mfK}$. The representations for $\bg_{1}$, $\bg_{2}$ and $h_{2}^k$ in \eqref{eq:proposition-solution-g1}, \eqref{eqn:proposition-solution-g2} and \eqref{eqn:proposition-solution-h2k}, respectively, together with the form of $\nujst$ in \eqref{eq:Trader-j-Ansatz-Structure}, provides us with a candidate for the optimal control in the population limit. It only remains to ensure that this optimal control is indeed admissible, i.e., $\nujst\in\A^j$. The following theorem provides sufficient conditions  for this to hold.

\begin{theorem} \label{thm:Final-Solution-Statement}
	Let us assume that $\bg_1\in\HT$.
	Then the optimality equation~\eqref{eq:Optimal-FBSDE-System} admits the solution
	\begin{equation} \label{eq:Solution-Structure-Statement}
\nujst_t = \nubarkst_t
+ \tfrac{1}{2 a_k} \, h_{2,t}^k
\,( \qj{\nujst}_t - \qbark{\nubarkst}_t )
			 \;,
	\end{equation}
and the mean-field trading rate process $\bnubarst = \left( \nubarkst \right)_{k\in\mfK}$ may be written
	\begin{equation}
\label{eqn:meanfield-opt-trading-rate}
		\bnubarst_t = \bg_{1,t} + \bg_{2,t} \, \bqbar_t^{\bnubarst},
	\end{equation}
	where $\bg_{1}$, $\bg_{2}$ and $h_{2}^k$ are the functions defined in Theorem~\ref{thm:Solution-g1-g2-h2}, and the mean-field inventory process $\bqbar^{\bnubarst}=(\qbark{\nubarkst})_{k\in\mfK}$ is
\[
\bqbar^{\bnubar}_t = \bmbar + \int_0^t \bnubarst_u \, du\,.
\]

Moreover, the collection of proposed optimal solutions satisfies
	\begin{equation}
		\nujst = \argmax_{\omega \in \A^j} \Hbar^{\bnubarst}_j(\omega)			
	\end{equation}
	for all $j\in\mfN$.
\end{theorem}
\begin{proof}
See \ref{sec:Proof-Final-Solution-Statement}.
\end{proof}

Theorem~\ref{thm:Final-Solution-Statement} guarantees, under the technical assumption that $\bg_{1}\in \HT$, our proposed solution forms a Nash-equilibrium for the limiting mean-field game. Moreover, Theorem~\eqref{thm:FBSDE-Optimality-Condition} guarantees that the solution is unique up to $\P\times\mu$ null sets.  The condition $\bg_{1}\in \HT$ holds for the class of models presented in Sections~\ref{sec:Example-Model-Subsection} and \ref{sec:Numerical-Experiments}.
While these models are not exhaustive, they provide an instructive class to study.

\subsection{Properties of the Optimal Control}
\label{sec:Properties-of-Optimal}

The optimal solution provided in Theorem~\eqref{thm:Final-Solution-Statement} admits many interesting properties. Firstly, the mean-field trading rate in \eqref{eqn:meanfield-opt-trading-rate} contains two parts: (i) a `risk control' portion $\bg_{2,t}\,\bqbar_t^{\bnubarst}$, which is independent of the dynamics of the asset price process; and (ii) an `alpha trading' or statistical arbitrage portion $\bg_{1}$.

The `risk control' portion ($\bg_{2,t} \, \bqbar_t^{\bnubarst}$) survives even when $A^k=0\;\forall k\in\mfK$, i.e., the midprice process subtracted from total order-flow is a martingale and induces interactions between the various sub-populations due to the their permanent impact.
It can be shown through numerical examples that this function scales with the parameter matrix $\phi$ and $\Psi$ to make agents liquidate their inventories faster when either $\phi$ or $\Psi$ become large, thereby controlling the risk agents take while trading.

{In the `alpha trading' portion ($\bg_{1,t}$), agents adjust their trading based on a weighted average of $\bAhat$, the estimated drift of the asset price for all agents. The weighting process $\bmcE$ encodes both information about the `risk' portion of the algorithm, $\bg_{2}$, as well as information about all other agent's measures through the process $\bZ$, which implicitly appears through the dynamics of $\bmcE$. The weighting function compensates for the differing models agents use for the asset price, and adjusts the individual trading rates to account for the price impact due to `alpha trading' of all other agents.
}

The Nash equilibrium, provided in Theorem \ref{thm:Final-Solution-Statement}, resembles the one obtained in \cite{casgrain2018algorithmic}, with the main differences lying in the expression for the value of the function $\bg_1$. The differences are important and reveal themselves in two ways.

First, here, we have a stochastic weighting process $\bmcE$ defined by the SDE~\eqref{eq:proposition-mcE} which replaces the deterministic time-ordered exponential function present \cite{casgrain2018algorithmic}. In fact, we can view $\bmcE$ as the natural extension of the time-ordered exponential appearing in \cite{casgrain2018algorithmic} to the case of stochastic processes. Second, to determine the correction to trading, rather than weighting a single estimate of future alpha as in \cite{casgrain2018algorithmic}, all posterior estimated alphas' $\Ahat^k$ under all measures $\Pk$, $k\in\mfK$, play a role. Finally, when $\Pk=\P^{k^\prime}$ for all $k,k^{\prime}\in\mfK$, the optimal controls in Theorem \ref{thm:Final-Solution-Statement} match the one presented in \cite{casgrain2018algorithmic}.

Thus far, we discussed the optimal mean-field strategy. The individual agents' trading rates also admit an interesting structure. An arbitrary agent trades at their own sub-population mean-field rate $\nubark$ plus a correction term proportional to the difference between their individual inventory and the mean-field inventory: $(\qj{\nujst}-\qbark{\nubark})$. This difference can be solved for in terms of the difference between the initial inventory of the agent and its sub-population prior mean:
\begin{equation}
	(\qj{\nujst}_t-\qbark{\nubarkst}_t) =
	(\mfQ_0^j - \mbar_k) \, e^{\int_0^t h_{2,u}^k \, du},
\end{equation}
where $h_{2,t}^k \leq 0$ for all $\tT$. Therefore, the difference in inventories shrinks towards zero at a deterministic rate, and agents are consistently drawing their inventories closer to their sub-population's mean-field. As time elapses, all agents in a sub-population resemble that of their sub-population's mean-field.

\section{The \texorpdfstring{$\epsilon$}{Epsilon}-Nash Equilibrium Property}
\label{sec:epsilon-Nash}

In Section~\ref{sec:Solving-The-MFG}, we solve the stochastic game in the infinite population limit, and provide an exact representation of each agent's control at the Nash-equilibrium. One important question to ask is how the optimal MFG strategy performs in a finite-population game. We study the properties of the limiting strategy in the finite game by looking at how close the collection of limiting strategies, defined in Theorem~\eqref{thm:Final-Solution-Statement} is to the true Nash-equilibrium of a game with only $N$ agents.

Let us consider a finite game with $N$ players, as described in Section~\ref{sec:The-Model-Description}. Let us assume that each of the agents in this population use the strategy described in Theorem~\ref{thm:Final-Solution-Statement}. Each agent computes the process $\bnubar_t$ according to equation~\eqref{eqn:meanfield-opt-trading-rate}, and then uses these values to compute their own trading rates, $\nujst_t$, according to equation~\eqref{eq:Solution-Structure-Statement}. In the theorem that follows, we show that this collection of controls can serve as a quasi-Nash-equilibrium in a finite player game, provided that the population size is large enough.

\begin{theorem}[$\epsilon$-Nash equilibrium] \label{thm:Epsilon-Nash-Theorem}
	Consider the collection of objective functionals $\left\{ H_j \, : j \in \mfN \right\}$ defined in equation~\eqref{eq:Objective-Function-Definition} and the set of optimal mean-field controls $\{ \nujst \}_{j=1}^{N}$ defined in Theorem~\eqref{thm:Final-Solution-Statement}. Suppose that there exists a sequence $\left\{ \delta_N \right\}_{N=1}^{\infty}$ such that $\delta_N \rightarrow 0$ and
	\begin{equation}
		\left\lvert \tfrac{N_k^{\N}}{N} -  p_k \right\rvert = o(\delta_N)
	\end{equation}
	for all $k\in\mfK$, then
	\begin{equation}
		H_j(\nujst , \nu^{-j,\ast}) \leq
		\sup_{\nu \in \A} H_j(\nu , \nu^{-j,\ast}) \leq
		H_j(\nujst , \nu^{-j,\ast}) + o( \tfrac{1}{N} ) + o(\delta_N)
	\end{equation}
	for each $j\in\mfN$.
\end{theorem}
\begin{proof}
See \ref{sec:Proof-Epsilon-Nash-Theorem}.
\end{proof}

{Theorem~\ref{thm:Epsilon-Nash-Theorem} shows that for any $\epsilon>0$, there exists $N_{\epsilon}$ such that for all $N>N_{\epsilon}$ agent-$j$ may improve their performance by at most $\epsilon$ by unilaterally deviating away from $\nujst$. The statement of the theorem also reveals that the rate $N_{\epsilon}$ must be at least linear
in $\epsilon^{-1}$ and is dependent on the rate at which $\delta_N$ vanishes in the limit. This theorem effectively demonstrates that the mean-field equilibrium $\{ \nujst \}_{j=1}^{N}$ serves as a viable alternative to the true finite-game equilibrium, provided the population size is large enough.}

\section{An Example Model of Disagreement}
\label{sec:Example-Model-Subsection}

In this part, we provide an example model where the asset price process is modulated by a latent Markov chain similarly to that in~\cite{casgrain_jaimungal_2016}. In our model, we assume each sub-population disagrees on the distribution of initial value of the latent process, while they do agree on what the possible values of the latent state are, and agree on the transition rates between states. One can view this as all agents believing there are positive, neutral, and negative drift environments, but disagree on what is the current environment. We prove that the resulting optimal control presented in Theorem~\ref{thm:Final-Solution-Statement} exists and is well-defined, i.e., that $\bg_1\in\HT$, under this general model assumption.

To this end, assume that the asset price satisfies the SDE
\begin{equation} \label{eq:Example-Model}
	dS_t^{\bnubarN} = \left( \sum_{i=1}^J \alpha_t^{i} \, \1{\Theta_t = \theta_i}  + \
	\sum_{k^\prime\in\mfK} \lambda_{k,k^\prime} \, p_{k^\prime}^{\N}\, \nubar^{k^\prime,\N}_u
	\, \right) dt + \sigma dW_t
	\;,
\end{equation}
where $\Theta_t$ is a continuous-time Markov chain taking values in the set $\{\theta_i\}_{i\in\mfJ}$ ($\mfJ=\{1,\dots,J\}$) and where the processes $\alpha^i = (\alpha_t^i)_{\tT}$ are $\F$-predictable processes satisfying $\alpha^i \in \HT$ for all $i\in\mfJ$. In this model, agents across different sub-populations have different prior probabilities on the initial value of the latent process, so that under the measure $\Pk(\Theta_0=\theta_i)=\pi_0^{k,i}\in(0,1)$ with $\sum_{i\in\mfJ}\pi_0^{k,i}=1$. We assume that under each measure $\Pk$, the latent Markov chain $\Theta_t$ has the same infinitesimal generator matrix\footnote{The generator matrix $\bC \in \R^{J\times J}$ can be any matrix satisfying the conditions $\bC_{i,j}\geq 0$ for all $i\neq j\in\mfJ$ and $C_{i,i} =\sum_{j\neq i\in\mfJ} C_{i,j}$. $\bC$ defines the transition dynamics of the latent Markov chain $\Theta_t$ through the relation, $\Pk\left( \Theta_{t+h} = \theta_{i} \Big \lvert \Theta_t = \theta_j \right) = \left( e^{h \bC }\right)_{i,j}$, where $e^{h \bC}$ represents the matrix exponential.} $\bC$. Furthermore, we assume that $W$ is a stardard Brownian motion in each measure $\Pk$ and that $\sigma>0$ is constant. We also simplify the impact model, and assume all agents have the same impact term $\blambda$ in each measure, so that, in the notation of section~\ref{sec:The-Model-Description}, we have $\blambda_k=\blambda$ for all $k\in\mfK$.

This model may be interpreted as a case in which agents all agree on the dynamics of the asset price $S^{\bnubar}$ and the latent process $\Theta$ but disagree on the initial value of the latent process. The specification allows us to compute the expression for the processes $\{ \bZ_t^{\Pk} \}_{k\in\mfK}$, which are used to compute each agent's optimal strategy. With this model, we may compute the Radon-Nikodym derivative process $\bZ_t^{\Pk}$ for any measure $\Pk$.

\begin{proposition}
\label{prop:Z-Prior-Expression-Prop}
Fix $\Q=\Pk$ for some $k\in\mfK$. If the asset price dynamics follow the latent Markov chain model of equation~\eqref{eq:Example-Model}, then $\bZ_t^{\Pk}$, defined in Theorem~\ref{thm:Solution-g1-g2-h2}, may be expressed as
	\begin{equation}
\label{eqn:ZPk-for-simple-model}
		\bZ_t^{\Pk}
		=
		\sum_{j\in\mfJ} \bm{M}_j^k
		\, \Pk \left(  \Theta_0 = \theta_j \big \lvert \F_t \right),
	\end{equation}
where for each $j\in\mfJ$  we define the diagonal matrix $\,\bM_j^k = \text{diag} \left( \pi_0^{k^\prime,j} \Big{/} \pi_0^{k,j} \right)_{k^\prime \in \mfK}$.
\end{proposition}
\begin{proof}
See~\ref{sec:Proof-Z-Prior-Expression-Prop}.
\end{proof}
From expression \eqref{eqn:ZPk-for-simple-model}, it is clear that $\bZ^{\Pk}$ is almost surely bounded, since  $\Pk \left(  \Theta_0 = \theta_j \big \lvert \F_t \right)\in[0,1]$ and $\pi^{k,i}_0\in(0,1)$ for all $k\in\mfK$, $i\in\mfJ$. We use this fact in the proof of the following proposition.

\begin{proposition} \label{prop:Example-Admissibility}
	Suppose that the asset price process is given by Equation~\eqref{eq:Example-Model}, then the solution $\bg_1$ defined in Theorem~\ref{thm:Solution-g1-g2-h2} satisfies $\bg_1\in\HT$ and thus the results of Theorem~\ref{thm:Final-Solution-Statement} apply to the model described in this section.
\end{proposition}
\begin{proof}
See \ref{sec:Proof-Example-Admissibility}.
\end{proof}

Although we show that there exist models for which the mean-field optimal control presented in Theorem~\ref{thm:Solution-g1-g2-h2} is well defined, computing these controls presents us with another challenge. In particular, due to the complicated nature of the process $\mcE^{\Q}$, the conditional expected value appearing in the expression~\eqref{eq:proposition-solution-g1}, for obtaining $\bg_1$, is difficult to compute. In section~\ref{sec:Numerical-Experiments}, we address this issue by presenting a computational method to approximate such expressions.

\section{A Simulation-Based Computational Method}
\label{sec:Computational-Method}

For most non-trivial models, obtaining a closed-form expression for the solution to the BSDE~\eqref{eq:proposition-solution-g1} for $\bg_{1,t}$ proves to be very difficult. To overcome this difficulty, we present a simulation-based computational method to approximate solutions. We propose a Least-Square-Monte-Carlo (LSMC) based method, which closely resembles the methods used to approximate solutions of BSDEs, as in \cite{bender2012least} and \cite{gobet2005regression}. Unlike these two methods, however, we do not concern ourselves with the computation of the martingale portion of the BSDE~\eqref{eq:g1-BSDE-Q-measure}, since it is not required to compute $\bg_{1}$.

To this end, define the $M$-point uniform partition of the interval $[0,T]$, $\mcT := \{ t_m := m \times \Delta ,\, m = 0 , 1 ,\dots , M \}$ where $M$ is a positive integer and where $\Delta := T/M$ is the discretization interval . We aim to approximate the process $\bg_{1}$ over the partition $\mcT$ with a discrete-time stochastic process $\bgh_1 = \left\{ \bgh_{1,t_m} \right\}_{t_m\in\mcT}$, where each $\bgh_{1,t_m}\in\R^{K}$.

To derive an expression for $\bgh_1$, we first study the expression for $\bg_{1,t}$,
\begin{equation}
	\bg_{1,t} = \E^{\Q}\left[\int_t^T
			(\bmcE_{t}^{\Q})^{-1} \, \bmcE_u^{\Q}
			 \; \bAhat_u
			\, du \, \Big \lvert \F_t \right]
	\;
\end{equation}
at the points $t_m\in\mcT$. This expression may be written recursively over $\mcT$ as follows
\begin{equation}
	\bg_{1,t_m} = \E^{\Q}\left[\int_{t_m}^{t_{m+1}}
			(\bmcE_{t_m}^{\Q})^{-1} \, \bmcE_u^{\Q}
			 \;\bAhat_u
			\, du
			+
			(\bmcE_{t_m}^{\Q})^{-1} \, \bmcE_{t_{m+1}}^{\Q} \; \bg_{1,t_{m+1}}
			\, \Big \lvert \F_{t_m} \right]
			\;.
\end{equation}
Next, approximating the time-integral in the previous expression with its left end-point, we obtain the approximation
\begin{equation} \label{eq:g1-approx-1}
	\bg_{1,t_m} \approx
	\E^{\Q}\left[
	\bAhat_{t_m} \Delta
	+
	(\bmcE_{t_m}^{\Q})^{-1} \, \bmcE_{t_{m+1}}^{\Q} \; \bg_{1,t_{m+1}}
			\, \Big \lvert \F_{t_m} \right]
			\;.
\end{equation}

A further simplification follows by approximating the term $(\bmcE_{t_m}^{\Q})^{-1} \, \bmcE_{t_{m+1}}^{\Q}$ for small values of $\Delta$. Using the definition of $\bmcE^{\Q}$ in Equation~\eqref{eq:proposition-mcE}, we may factor $\bmcE^{\Q}$ as $\bmcE^{\Q}_t = \bmcEt^{\Q}_t \bZ_t^{\Q}$, where $\bmcEt^{\Q}_t$ is the solution the the matrix-valued SDE $d\bmcEt^{\Q}_t = \bmcEt^{\Q}_t \left(\bZ_t^{\Q} \, \bG_t \,(\bZ_t^{\Q})^{-1} \right)\, dt$ with initial condition $\bmcEt^{\Q}_0 = \bm{I}^{(K\times K)}$. For $\Delta\ll 1$, we freeze the process in parenthesis at their $t_m$ values, so that $d\bmcEt^{\Q}_t \approx \bmcEt^{\Q}_t \left(\bZ_{t_m}^{\Q} \, \bG_{t_m} \,(\bZ_{t_m}^{\Q})^{-1} \right)\, dt$ over each interval $[t_m,t_{m+1})$, resulting in
\begin{align}
	(\bmcEt_{t_m}^{\Q})^{-1} \, \bmcEt_{t_{m+1}}^{\Q} &\approx
	\exp\left\{ \bZ_{t_m}^{\Q} \, \bG_{t_m} \,(\bZ_{t_m}^{\Q})^{-1} \Delta \right\}
=
	 \bZ_{t_m}^{\Q} \exp\left\{ \bG_{t_m} \Delta \right\} (\bZ_{t_m}^{\Q})^{-1}
	 \;,
\end{align}
where $\exp$ represents matrix exponential.
By plugging in this last result into equation~\eqref{eq:g1-approx-1}, we obtain an approximation $\bgh_1$ for the process $\bg_1$ at $t_m$ as
\begin{equation} \label{eq:g1-approx-2}
	\bgh_{1,t_m} =
	\E^{\Q}\left[\left.
	\bAhat_{t_m} \Delta
	+
	\exp\left\{ \bG_{t_m} \Delta \right\}
	(\bZ_{t_m}^{\Q})^{-1}
	\,
	\bZ_{t_{m+1}}^{\Q}\; \bgh_{1,t_{m+1}}
	\, \right\lvert\, \F_{t_m} \right]
	\;.
\end{equation}

The final step in obtaining values of $\bgh_1$ is to approximate the conditional expected value in the rhs of equation~\eqref{eq:g1-approx-2}. As is often done, we project the conditional expectation onto a finite basis of stochastic processes. In particular, let the (vector-valued) stochastic process $\bY = \left( \bY_t\right)_{\tT}$, with $\bY_t \in \R^{L}$ where $L$ is some positive integer, and we write
\begin{equation}
	\E^{\Q}\left[\left.
	\bAhat_{t_m} \Delta
	+
	\exp\left\{ \bG_{t_m} \Delta \right\}
	(\bZ_{t_m}^{\Q})^{-1}
	\,
	\bZ_{t_{m+1}}^{\Q}\; \bgh_{1,t_{m+1}}
	\, \right \lvert \,\F_{t_m} \right]
\approx \left\langle \bY_{t_m} , \bbeta_{t_m} \right\rangle
\end{equation}
for some collection $\{\bbeta_t\}_{t\in\mcT}$, where each $\bbeta_{t}\in\R^{L \times K}$, and where the process $\bY$ can be chosen fairly arbitrarily. A common and sensible choice for $\bY$ is a finite basis expansion of the state processes of the problem (i.e. $S_t^{\nubar}$, $\bZ_t^{\Q}$, etc.) and combinations of them.


The algorithm then estimates the coefficients $\bbetah$  in a sequential manner. This is done by first simulating paths of $\bY_t$ forward over the time partition $\mcT$ using the measure $\Q$, and then proceeding backwards in time from the boundary condition, solving a least-square regression problem at each time step $t_m \in \mcT$ to obtain each of the coefficients $\bbetah$. The details of this algorithm are illustrated in Algorithm~\ref{algo:LSMC-Algorithm} below. Algorithm~\ref{algo:LSMC-Algorithm} is an application of the LSMC methods that already exist for BSDEs and we point the reader to \cite{bender2012least} and \cite{gobet2005regression} for more details on the convergence rates and error bounds.

\begin{figure}[H]
  \centering
  \begin{minipage}{.7\linewidth}
    \begin{algorithm}[H]
      \SetAlgoLined
      \KwData{Simulate $\mathfrak{M}_0$ paths of ($\bY_t$,$\bZ_t^{\Q}$,$\bAhat$) over $\mcT$ using measure $\Q$}
      Set $\bbetah_{t_M} = \bm{0}^{(L \times 2)}$ \\
      Set $\bgh_{1,t_M}(\bY) = \bm{0}^{(L \times 2)}$ \\
      \For{$m = M-1 , M-2 , \dots , 1$}{
      	Set
      \begin{align*}
      \bbetah_{t_m} = \arg\min_{\beta}
      	\sum_{n=1}^{\mathfrak{M}_0} \Big( & \left\langle \bY_{t_m}^n, \bbeta \right\rangle
      \\
      & -
      \left\{
      	\bAhat_{t_m}^n \Delta
		+
		\exp\left\{ \bG_{t_m} \Delta \right\}
		(\bZ_{t_m}^{\Q,n})^{-1}
		\,
		\bZ_{t_{m+1}}^{\Q,n}\; \bgh_{1,t_{m+1}} (\bY_{t_{m+1}}^n)
    \right\}
      	 \Big)^2
      \end{align*}
      Set $\bgh_{1,t_m}(\bY) = \left\langle \bY_{t_m}^n, \bbetah_{t_m} \right\rangle$
      }
      \caption{The LSMC algorithm used to approximate the value of the process $\bg_1$ given in Equation \eqref{eq:proposition-solution-g1}.}
      \label{algo:LSMC-Algorithm}
    \end{algorithm}
  \end{minipage}
\end{figure}

As the process $\bZ^{\Q}$ is defined as a diagonal matrix of Radon-Nikodym derivatives, it is possible to re-write conditional expected value over $\Q$ in equation~\eqref{eq:g1-approx-2} in an element-wise fashion as
\begin{equation} \label{eq:g1-approx-3}
	\hg^k_{1,t_{m+1}}
	=
	\Ahat_{t_m}^k \Delta
	+
	\sum_{k^\prime \in \mfK}
	\left(\exp\left\{ \bG_{t_m} \Delta \right\}\right)_{k,k^{\prime}}
	\E^{\P^{k^\prime}}\left[ \left.
\hg^{k^\prime}_{1,t_{m+1}} \, \right \lvert \,\F_{t_m} \right],\qquad \forall k\in\mfK,
\end{equation}
where $\left( \; \cdot \;\right)_{k,k^{\prime}}$ represents element $(k,k^{\prime})$ of the matrix. The above representation allows one to modify Algorithm~\ref{algo:LSMC-Algorithm} such that it eliminates the dependence on the process $\bZ^{\Q}$ in the LSMC procedure, but at the cost of having to simulate the basis process $\bY$ across all measures $\{\Pk\}_{k\in\mfK}$. We find that in examples where simulating the process $\bZ_t^{\Q}$ is straightforward, this is much less efficient than Algorithm~\ref{algo:LSMC-Algorithm} due to the need of simulating and storing  $K$ copies of the process $\bY$. In cases where $\bZ^{\Q}$ is intractable, however, this modification may be a viable alternative for computing $\bgh_1$.

Equation~\eqref{eq:g1-approx-3} also provides additional insight into how the optimal policy is trading. As pointed out in Section~\ref{sec:Properties-of-Optimal}, the process $\bg_1$ represents the `statistical arbitrage' portion of the agent's optimal trading strategy. {Equation~\eqref{eq:g1-approx-3} further reveals that, over one step, agent's of type-$k$ trade proportionally to the sum of their best estimate of the asset's drift ($\Ahat^k_{t_m} \Delta$) and a weighted average of the expected end of period `alpha' from all sub-populations.
Hence, agents trade based on expected exogenous price movements plus what they anticipate other traders' actions to have on price.
The weights generated by the matrix $\exp\left\{ \bG_{t_m} \Delta \right\}$ serve to risk adjust the agent's own alpha trading and to adjust for the impact of other agents based on the scale of their market impacts.
}

\section{Numerical Experiments}
\label{sec:Numerical-Experiments}

This section showcases numerical experiments resulting from a particular model of differing beliefs. We first assess the performance of the LSMC algorithm presented in Section~\ref{sec:Computational-Method} by comparing, in the case of equal beliefs, to the analytical results in \cite{casgrain2018algorithmic}. The algorithm is then used to approximate and simulate a finite collection of agents trading at the mean-field Nash-equilibrium when the agents have differing beliefs.

For the remainder of the section, we assume the asset price process follows a linear mean-reverting model described in Section~\ref{sec:Example-Model-Subsection} with $K=2$ sub-populations. Define the un-impacted asset price process $F=(F_t)_{\tT}$ to be the solution to the SDE
\begin{equation} \label{eq:Unimpacted-Price-Example-Def}
	dF_t = \kappa \left( \Theta_t - F_t \right) \, dt + \sigma \,dW_t\,,
\end{equation}
where $\kappa, \sigma > 0$, $W=(W_t)_{\tT}$ is a Wiener process in both measures $\P^1$ and $\P^2$, and $\Theta=(\Theta_t)_{\tT}$ is a latent Markov chain with generator matrix $\bm{C}$ which can take one of two values in the set $\{\theta_1,\theta_2\}$. The asset price process including the price impact is then defined as having the dynamics
\begin{equation*}
	dS_t^{\bnubarN} = dF_t + \left( \lambda_1 \, p_1^{\N} \, \nubar_t^{1,\N}  +  \lambda_2 \, p_2^{\N} \, \nubar_t^{2,\N} \right) \, dt
\end{equation*}
with $\lambda_1,\lambda_2 > 0$ and $(p_k^{\N})_{k\in\mfK}$ defined in Section~\ref{sec:The-Model-Description}. We assume sub-population 1 believes the initial value of $\Theta_0$ has distribution $\bpi_0^1$, while sub-population 2 believes the initial value has distribution $\bpi_0^2$. The dynamics of the asset price process causes it to mean-revert towards the value of $\Theta_t$, which may change over the course of the trading period $[0,T]$. Furthermore, this model falls into the class of models described in Section~\ref{sec:Example-Model-Subsection}, which guarantees that the mean-field optimal solution from Theorem~\ref{thm:Final-Solution-Statement} exists and is well defined.

\subsection{The LSMC Algorithm}
\label{sec:LSMC-Examples}

To assess the LSMC algorithm described in Section~\ref{sec:Computational-Method}, we choose a special, non-trivial, case where $\bg_1$ can be computed in closed form. The case we study is when $\P^1=\P^2=\P$. This reduces the market model to one where all agents agree on the dynamics of the asset price process. We may then assess the accuracy of the approximation by comparing the results produced by the LSMC algorithm to the closed-form solution of the optimal control in \cite{casgrain2018algorithmic}.

For this particular experiment, we use the model presented in the previous section, but where the prior on the initial states of the latent process is the same for all agents. The two sub-populations of agents may, however, differ in their parameter triplet $(\Psi_k,\phi_k,a_k)$. For the experiments we use the parameters in Table~\ref{tab:sim-player-params}. The parameters chosen for this experiment match the parameters used in the simulations in Section 5 of \cite{casgrain2018algorithmic}. Due to the large value of the parameter $\Psi_k$, agents in both sub-populations are incentivized to fully liquidate their inventory positions before the end of the trading horizon. The risk-aversion parameter $\phi_k$ is $10$ times larger in sub-population 2 than in sub-population 1. This can be interpreted as a model in which agents in sub-population 2 are averse to holding any inventory and are intent on liquidating their inventories as quickly as possible, while agents in sub-population 1 do not feel such urgency and are more open to trading on alpha.
\begin{table}[htbp]
  \centering
    \begin{tabular}{r|rrll}
    \multicolumn{1}{l|}{$k$} & \multicolumn{1}{l}{$N_k$} & \multicolumn{1}{l}{$\Psi_k$} & $\phi_k$ & $a_k$ \bigstrut[b]\\
	\cline{2-5}
	1    & 20    & 10   & $10^{-2}$ & $10^{-4}$ \bigstrut[t]\\
    2    & 10    & 10   & $10^{-6}$ & $10^{-4}$ \\
    \end{tabular}%
  \caption{Population and impact parameters for the two sub-populations of agents.}
  \label{tab:sim-player-params}%
\end{table}%

We set $T=1$ to be the trading horizon for the model. The asset price process follows the Markov modulated Ornstein-Uhlenbeck dynamics in \eqref{eq:Unimpacted-Price-Example-Def},  with parameters provided in Table~\ref{tbl:sim-asset-params}. Table~\ref{tbl:sim-asset-params} also defines the parameters for the dynamics of the latent process $\Theta_t$. $\Theta_t$ is defined so that the asset price process either mean reverts to $\theta_1=4.95$ or $\theta_2=5.05$, depending on the state of $\Theta$. In this particular experiment, we set the distribution of $\Theta_0$ so that there is an equal chance of starting in each of the states. We also choose an asymmetric generator matrix so that the latent process is twice as likely to spend time in state 1 than state 2.

\begin{table}[h!]
  \centering
  \begin{tabular}{lllll}
  $\bm{\pi}_0 = \left(\begin{smallmatrix} 0.5 \\ 0.5 \end{smallmatrix}\right)$, &
  $\bm C= \left[ \begin{smallmatrix} -1 & \,\phantom{-}1 \\ \phantom{-}2 &\, -2  \end{smallmatrix} \right]$, &
  $\bm{\theta} = \{4.95, 5.05\}$, & \\
  $\kappa = 5.4$, &
  $\sigma = 0.185$, & \text{ and }
  $\lambda_{k} = 10^{-3}$
  \end{tabular}%
  \caption{The parameters used for the asset price dynamics and for the latent process.}
  \label{tbl:sim-asset-params}
\end{table}

We run Algorithm~\ref{algo:LSMC-Algorithm} using $10^4$ simulated paths over a partition of size $3600$ over the interal $[0,T]$ and compare the results from the LSMC algorithm applied to these simulated paths to the closed form solution for $\bg_1$ in \cite{casgrain2018algorithmic}. In this particular case, we set the basis process, $\bY_t$, to be a second-order monomial expansion of $(S_t^{\bnubar},\bpi_t)$ with product terms included, where we define $\bpi^i_t = \P\left(\Theta_t = \theta_i \big\lvert \F_t \right)$ for $i=1,2$. For details on how to compute such conditional probabilities, see Section~\ref{sec:Filtering-and-Smoothing-Appendix} and \cite[Section 3]{casgrain_jaimungal_2016}.}

\begin{figure}[H]
    \centering
    \includegraphics[width=0.6\textwidth]{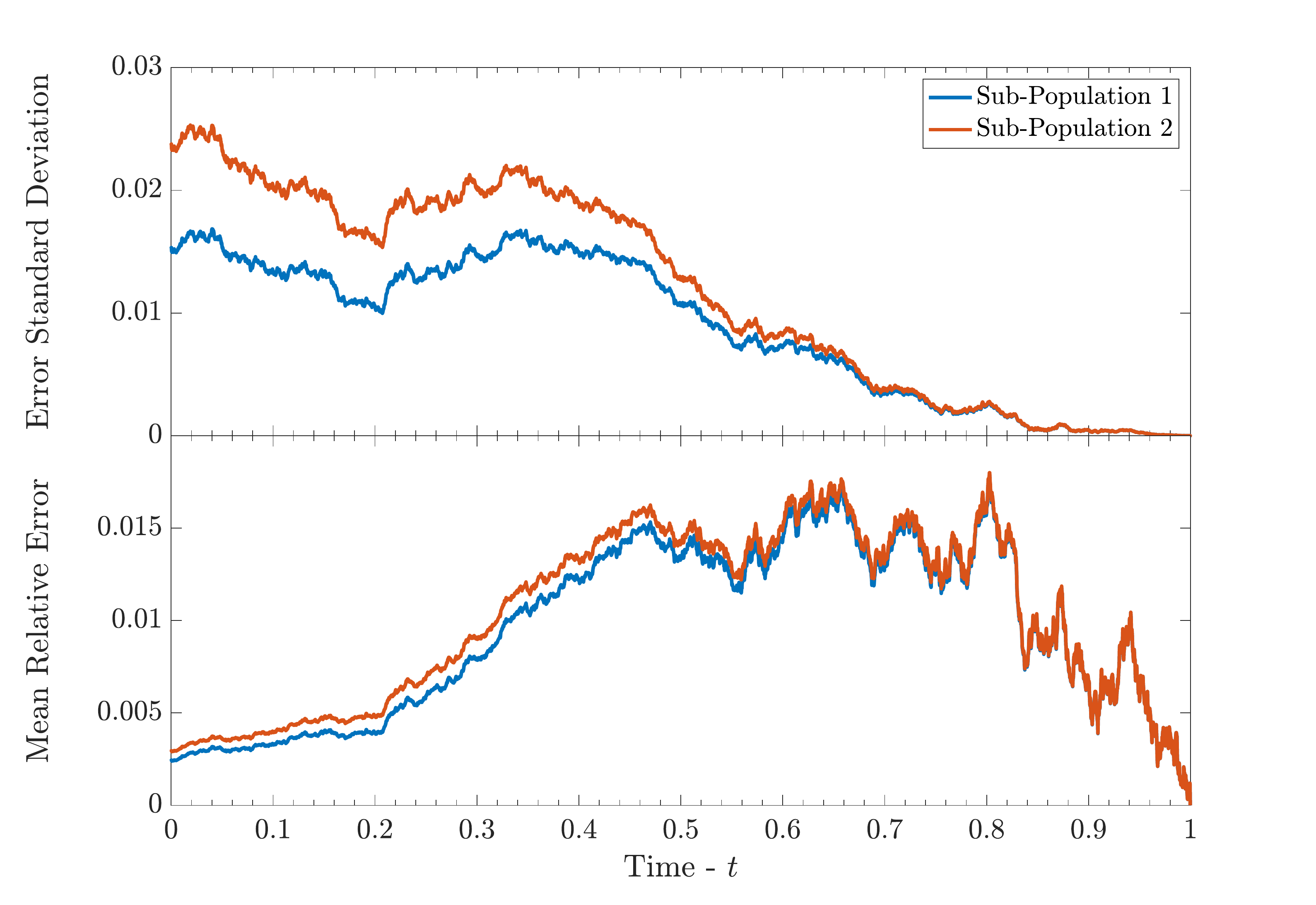}
    \caption{Error plots for the LSMC algorithm described in Section~\ref{sec:Computational-Method}. In these plots, we compare the value of the LSMC estimate, $\bgh_1$, with the true value of $\bg_1$, in a special case where we can compute $\bg_1$ in closed form. The upper panel shows the standard deviation of the error, $\text{SD}\left(\bgh_{1,t}^k - \bg_{1,t}^k\right)$ computed over $10^4$ simulations, for each $k=1,2$ over the interval $[0,T]$. The lower panel, plots the quantity $\E \left[ \lvert \bg_{1,t}^k - \bgh_{1,t} \rvert \right] \,\big{/}\, \E \left[ \lvert \bg_{1,t}^k \rvert \right]$ computed over $10^4$ simulations, and provides a measure of relative error.}
    \label{fig:LSMC-Error-Plot}
\end{figure}

Figure~\ref{fig:LSMC-Error-Plot} shows that the LSMC algorithm performs well and with a high level of accuracy with this particular model. In particular, from the lower panel, we see  that the largest relative error is about $1.5\%$, meaning that the error is reliably no more than $1.5\%$ of the absolute size of $\bg_{1}$.  We have also observed, as elsewhere in the LSMC literature such as in \cite{letourneau2016improved} and \cite{wang2009pricing}, that randomizing the initial value of the state process, $(S_0,\bpi_0)$, for the forward simulation portion of the algorithm significantly improves the estimates. Furthermore, the errors reported in this figure appear to be consistent across a wide variety of model parameters.

For the more general case in which there are different measures assigned to each population, we set $\Q=\P^1$ and we enlarge the basis process $\bY_t$ to be a monomial expansion of the forward state process $(S_t^{\nubar},\{\bpi_t^k \}_{k\in\mfK},\{\bZ_t^{\P^1,k}\}_{k\in\mfK})$. Expansions with respect to different bases, such as Laguerre or Hermite polynomials, are also possible, however, in our experiments, we find the monomial basis expansion performs well enough.

\subsection{A Simulation of the Market}
\label{sec:Market-Simulations}

In this section, we simulate the full market with agents of differing beliefs disagreement. The example continues to use the model in Section~\ref{sec:LSMC-Examples} with the parameters in Table \ref{tab:sim-player-params} and \ref{tbl:sim-asset-params}, with the exception that the distributions on $\Theta_0$ now differs across each of the sub-populations. In particular, we assume agent's in sub-population 1 believe that prior distribution over initial states is $\bm{\pi}_0^{1} = \left(\begin{smallmatrix} 0.1 \\ 0.9 \end{smallmatrix}\right)$, while the sub-population 2 believe it is $\bm{\pi}_0^{2} = \left(\begin{smallmatrix} 0.9 \\ 0.1 \end{smallmatrix}\right)$. In other words, sub-population 1 believes that the latent process will much more likely begin in the higher state, while sub-population 2 assumes the reverse. In the simulation, we also assume the starting inventory of agents in sub-population $k$ has distribution $\mfQ_0^j \sim \mathcal{N}(\bar{\mu}_k,\bar{\sigma})$, where we set $\bar{\mu}_1 = 100$, $\bar{\mu}_2 = 0$ and $\sigma = 50$. The rationale is so that the risk-averse sub-population $1$ begins the trading period long $100$ shares on average, while agents in sub-population $2$ begin the trading period with zero shares on average. Over the course of the simulation, we fix the path of the latent process to begin in the upper state and then jump down to the lower state at $t=0.5$. To compute the trading strategy of each participating agent, we use the LSMC method from the preceding section to approximate the value of $\bg_1$ and then use this value in Theorem~\ref{thm:Final-Solution-Statement} to determine the optimal trading rate of the fictitious mean-field and then each individual agent. At each time step, we compute the basis process $\bY_t$ by using a fifth-order polynomial expansion of the state process $(S_t^{\nubar},\{\bpi_t^k \}_{k\in\mfK},\{\bZ_t^{\P^1,k}\}_{k\in\mfK})$, and use the coefficients obtained by the LSMC algorithm to obtain an approximation for $\bg_{1,t}$. Computing the values of $\bpi_t^k$ and $\bZ_t^{\P^1,k}$ requires the computation of a collection of posterior probabilities at each time step. To do this, we make use of the filtering and smoothing equations which are detailed in Appendix~\ref{sec:Filtering-and-Smoothing-Appendix}.

\begin{figure}[h]
    \centering
    \begin{subfigure}[t]{0.41\textwidth}
        \centering
        \includegraphics[width=\textwidth]{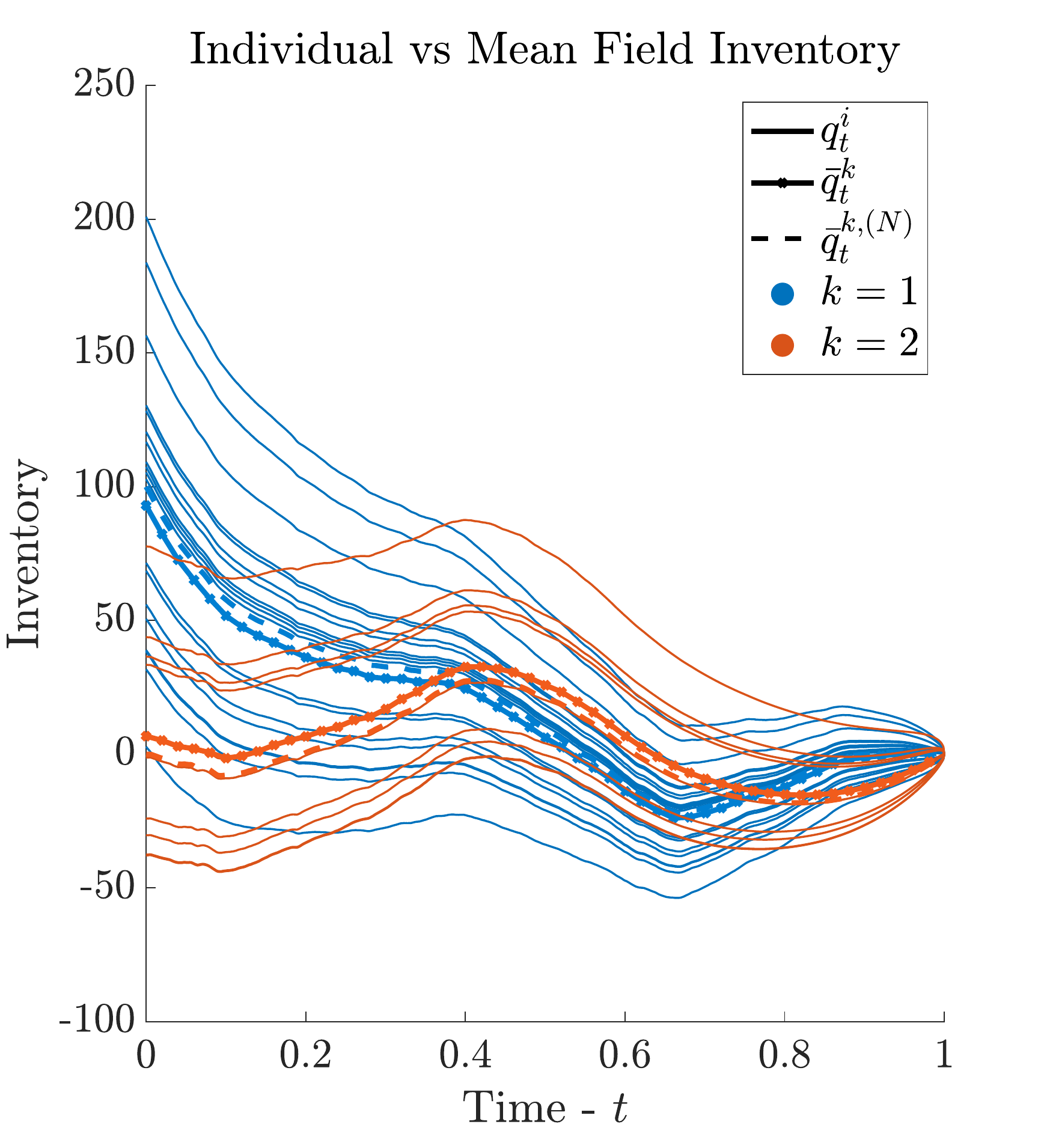}
    \end{subfigure}
    ~
    \begin{subfigure}[t]{0.41\textwidth}
        \centering
        \includegraphics[width=\textwidth]{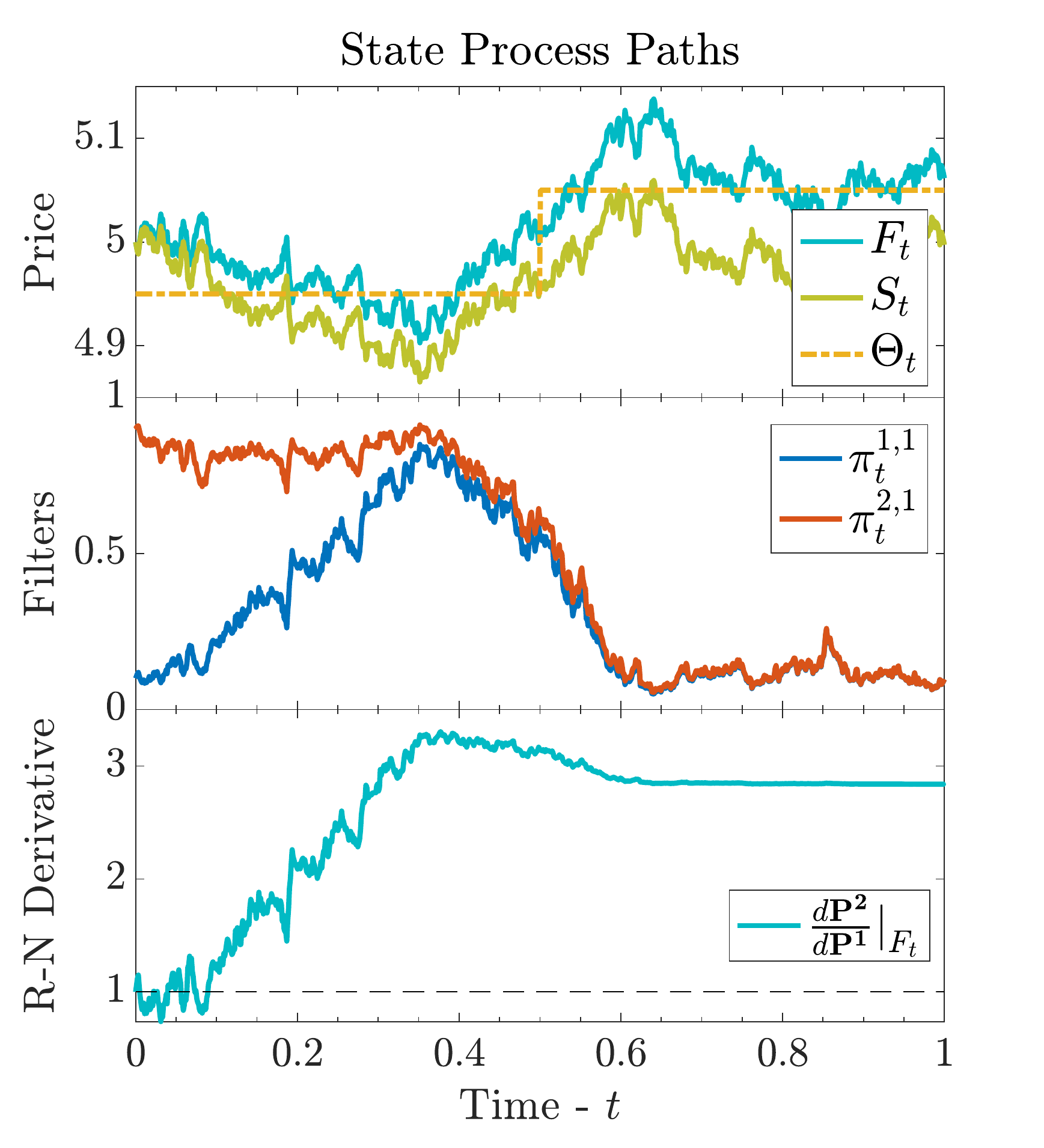}
    \end{subfigure}
    \caption{
    State processes from a single simulated scenario of the market. \textsl{Left panel}: inventory path process from all agents (sub-populations separated by color), the sub-population mean-field inventory process $\bqbar^k$, and the sub-population empirical mean inventory $\bqbar^{k,\N}$. \textsl{Right panel}: (\textit{top}) the unimpacted $F$ and impacted $S$ asset price processes, and the latent Markov chain $\Theta$; (\textit{middle}) sub-population filters $\pi_t^{k,j}=\Pk\left( \Theta_t = \theta_j \big\lvert \F_t \right)$ for the latent process state; (\textit{bottom}) the Radon-Nikodym process $Z_t^{\P^{1,2}} = \frac{d\P^2}{d\P^1}\big\lvert_{\F_t}$.}
    \label{fig:Market-Simulation-Paths}
\end{figure}

Figure~\ref{fig:Market-Simulation-Paths} shows one sample path of the simulation of all agents. The figure demonstrates a number of path-wise properties of the trading algorithm and of the beliefs of each of the sub-populations of agents. Firstly,  the left panel shows that the agents inventory paths differ significantly between the sub-populations. As mentioned in Section~\ref{sec:LSMC-Examples}, and resulting from the population parameters in Table \ref{tab:sim-player-params}, sub-population $1$ is far more risk-averse than sub-population $2$. This is reflected in the path-wise variance of their inventory.

Agents in sub-population $1$ begin long the asset on average. As these agents are risk-averse, their main concern is to unwind their position quickly. They are, however, conscious of their own expectations of the future path of the asset price as well as the expectations of sub-population $2$, which they use to adjust the rate at which their inventory is liquidated. This last effect can be seen through the variations of the inventory paths of sub-population 2 in Figure~\ref{fig:Market-Simulation-Paths}.

Agents in sub-population 2 are instead concerned with profiting from statistical arbitrage. They begin the trading period by incorrectly assigning a 90\% probability that the latent process is in the upper state. Because of this, they expect the asset price to mean-revert downwards slightly, so they begin by taking a slight short position in the asset over the time period $t\in[0,0.15]$. By time $t=0.15$, the asset price has approximately reached the lower mean reverting level. The agent expects that the asset price will now be reverting upwards in the long run, since it expects the state of the latent process to switch, which would cause the price to begin reverting upwards. Because of this, the agent begins reverting their short position into a long position in the asset over the course of the time period $t\in[0.15,0.4]$. The asset price trajectory indeed reverts upwards, and these agents gradually update their posterior state distributions. By time $t=0.4$, agents from group 2 are now confident that the latent process is in the upper state. Moreover, using the same train of logic as before, it expects the price to mean revert downwards in the long run, due to an expected switch in the latent process. Thus it gradually shifts to a short position and repeats the same process. The magnitude of the long and short positions for sub-population 2 decrease as the end of the trading period approaches. This is due to the fact that the agent is highly insentivized to completely liquidate their inventory before time $t=1$, and therefore reduces their absolute exposure so that it is easier to completely liquidate their inventory.

From the ceter-right part of Figure~\ref{fig:Market-Simulation-Paths}, we also see that the posterior distribution over latent states for each group converge to one another as time progresses. This is since, although their priors are different, the agents are able collect information so that the effect of the priors on the final posterior computation is negligible by a certain time. Furthermore, as was pointed out in the discussion following equation~\eqref{eq:g1-approx-3}, the strategies of agents from different sub-populations feed into one another. This causes agents from different sub-population to move synchronously with respect to one another, as seen in the left of Figure~\ref{fig:Market-Simulation-Paths}, where the upwards and downward variations in agent's strategies happen simultaneously.

The actions of agents from both sub-populations demonstrate that the optimal control incorporates the beliefs of all agents and weighs them against their own. The filter paths in the middle right panel of Figure~\ref{fig:Market-Simulation-Paths} show how both agents eventually learn the true value of the latent state with high confidence. And this occurs by observing the paths of the price process only, even if their beliefs on its initial state are incorrect. The Radon-Nikodym derivative path in the bottom middle panel of Figure~\ref{fig:Market-Simulation-Paths} provides a sense of how far apart are the measures for sub-populations $1$ and $2$. This process varies significantly over the course of the trading period since agents are constantly updating their estimate of the latent price process by observing order-flow and the price paths. The variation in this process also demonstrates there is a non-trivial interdependence between the actions of each agent and the beliefs of all other agents.

\subsection{The Effect of Disagreement on Markets}

Using the same latent Markov model as in Section~\ref{sec:Market-Simulations}, we study the predicted behaviour of market prices and of market participant as we vary the degree of disagreement across sub-populations. We investigate the effect of disagreement on both the volatility of market prices and the total trading turnover of market participants.

We assume $K=2$ sub-populations of equal size, each with $N_k=30$ agents. Each of sub-population has identical preferences, but differ in their beliefs of the market and set the agents' preference parameters to $\Psi_k=a_k=10^{-4}$, $\phi_k=5 \times 10^{-4}$ for $k=1,2$, and $\Psi_k=a_k$ (so that agents are not necessarily forced to arrive at time $T$ with zero inventory). The initial inventory positions of agents are drawn from $\mfQ_0^j \sim \mathcal{N}( 0, \bar{\sigma})$ for all $j\in\mfN$, with $\bar{\sigma}=50$.

The two-state latent Markov process $\Theta_t$ has generator matrix $\bm C= \boldmath{0}$,
and hence $\Theta_t=\Theta_0$ for all $t\in[0,T]$, however, $\Theta_0$ is random and inaccesible to agents. Table~\ref{tbl:sim-asset-params-disagreement} lists the assumed parameters of the mean-reverting asset price pocess, as well as the latent process.

\begin{table}[h!]
  \centering
  \begin{tabular}{llllll}
  $S_0 = 5$, &
  $\bm{\theta} = \{4.95, 5.05\}$, &
  $\kappa = 5.4$, &
  $\sigma = 0.14$,  \text{ and }
  $\lambda_{1} = \lambda_{2} = 10^{-1}$.
  \end{tabular}%
  \caption{The parameters used for the asset price dynamics to generate Figure~\ref{fig:Market-Variance-Disagreement}.}
  \label{tbl:sim-asset-params-disagreement}
\end{table}

To introduce disagreement into this setup, we assume that sub-populations have different prior beliefs on the distribution of $\Theta_0$. In particular, we assume $\pi_0^1 = \left(\begin{smallmatrix} 0.5 + \Delta \pi_0 \\ 0.5 - \Delta \pi_0 \end{smallmatrix}\right)$ and $\pi_0^2 = \left(\begin{smallmatrix} 0.5 - \Delta \pi_0 \\ 0.5 + \Delta \pi_0 \end{smallmatrix}\right)$, where $\Delta \pi_0 \in [0,0.5)$ quantifies the level of disagreement across the two sub-populations. In simulations, we assume the true probability distribution of the latent process is $\P(\Theta_0=4.95)=\P(\Theta_0=5.05)=0.5$.

\begin{figure}[h]
    \hspace{-1.2cm}
    \begin{subfigure}[t]{1.1\textwidth}
        \centering
        \includegraphics[width=\textwidth]{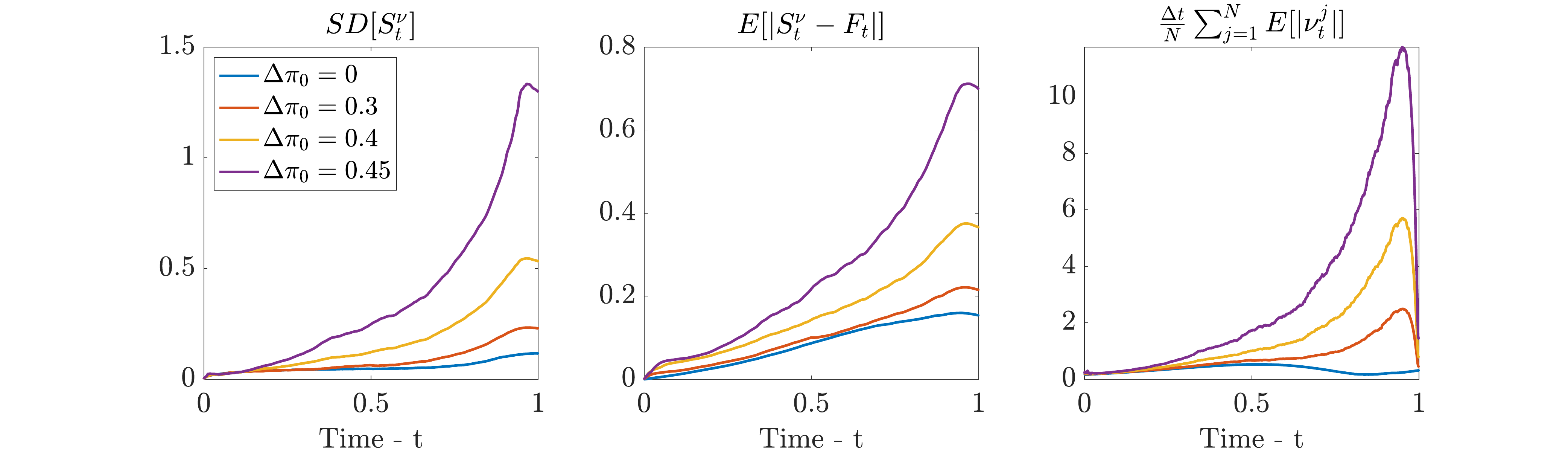}
    \end{subfigure}
    \caption{Estimated statistics of the simulated market as the degree of disagreement $\Delta \pi_0$ varies.
    \textit{left panel}: standard deviation of the asset price. \textit{center panel}: absolute deviation of asset price from un-impacted price. \textit{right panel}: average absolute trading rate.
}
    \label{fig:Market-Variance-Disagreement}
\end{figure}

Figure~\ref{fig:Market-Variance-Disagreement} shows various statistics (standard deviation of price, average absolute deviation from the un-impacted price, and average absolute trading rate) of trading activity within the market resulting from $10^4$ simulations. All three panels show a unilateral increase in all of the plotted statistics  as the level of disagreement increases. In particular, the right panel shows that trading volume increases, driving up the standard deviation of the asset price process (as see in the left panel), and driving up the net impact of trading as shown in the center panel.

Extrapolating from the results of these experiments, we can conclude that an increase in disagreement amongst a population of agents appears to increase market volume and increase asset price volatility. These observations are consistent with those seen in~\cite{bayraktar2017mini}, who also observe an increase in market activity as disagreement increases in markets.

\section{Conclusion}

This paper introduced a stochastic game for a market in which sub-populations of agents have different risk-preferences and beliefs on the model for the asset price process. By taking the infinite population limit of the model, we obtained a more tractable mean-field game (MFG) model for the market. By using tools from convex analysis we provide an FBSDE characterization of the optimal control of each agent and thus the Nash-equlibirum of the MFG. This FBSDE is high dimensional, and non-standard as the martingale components for each dimension are martingales across different probability measures. Through some change-of-measure techniques we manage to obtain a solution to this FBSDE system and for the collection of mean-field optimal controls. We also demonstrated that the MFG optimal control satisfies the $\epsilon$-Nash property, which implies that the limiting Nash-Equilibrium can be arbitrarily close to the Nash-equilibrium in the finite player game as long as the population size is large enough. Lastly, we provide a LSMC approximation to the MFG optimal control, and use it to study example simulations of markets near their Nash-equilibrium. {In a simulation setting, increasing disagreement among market participants appears to increase price volatility, price deviation from the un-impacted market price, and trading volume.}


\newpage

\appendix
\footnotesize

\section{Proofs}

\subsection{Proof of Lemma~\ref{thm:Convex-Lemma}}
\label{sec:Proof-Convex-Lemma}

\begin{proof}
	To show that the claim holds, we need to show that for any $\rho\in(0,1)$,
	\begin{equation} \label{eq:Convex-Lemma-1}
		\Hbar^{\bnubar}_j( \rho \nu + (1-\rho) \omega )
		- \rho \Hbar^{\bnubar}_j(\nu) - (1-\rho) \Hbar^{\bnubar}_j(\omega)
		> 0
	\end{equation}
	for all $\nu,\omega \in \A^j$ where $\nu_t=\omega_t$ at most on $\P\times\mu$ null sets. By noting that
	\begin{equation}
		\qj{\rho \nu + (1-\rho) \omega}_t  =
		\rho \, \qj{\nu}_t
		+
		(1-\rho) \,\qj{\omega}_t
		\;,
	\end{equation}

	we may compute the difference~\eqref{eq:Convex-Lemma-1} using the representation~\eqref{eq:Objective-Function-Alternative-Representation} of $\Hbar^{\bnubar}_j$ to obtain,
	\begin{align*}
	\text{LHS of \eqref{eq:Convex-Lemma-1}}
		=\E^\Pk\Bigg[ &\int_0^T\Big\{
		 \rho \smallqfm{\nu}^\T \bm\Gamma_k \smallqfm{\nu}  +
		(1-\rho) \smallqfm{\omega}^\T \bm\Gamma_k \smallqfm{\omega} \\
		&\quad\qquad-
		\left( \rho \smallqfm{\nu} + (1-\rho) \smallqfm{\omega} \right)^\T
		\bm\Gamma_k
		\left( \rho \smallqfm{\nu} + (1-\rho) \smallqfm{\omega} \right)\Big\} dt \Bigg]
		\\
		\text{(completing the square)}		
		= \E^\Pk \Bigg[ & \int_0^T \Big\{
		\rho\,(1-\rho)
		\left( \smallqfm{\nu} - \smallqfm{\omega} \right)^\T
		\bm\Gamma_k
		\left( \smallqfm{\nu} - \smallqfm{\omega} \right)
		\Big\} dt \Bigg]
		\;,
	\end{align*}
	where we define the matrix $\bGamma_k = \left( \begin{smallmatrix} a_k & \Psi_k \\ \Psi_k & \phi_k \end{smallmatrix}\right)$.

	By defining the terms $\Delta_t = \nu_t - \omega_t$, $q^{\Delta}_t=\qj{\nu}_t-\qj{\omega}_t$, we can expand the above expression to obtain
	\begin{equation} \label{eq:Convec-Lemma-2}
\text{LHS of \eqref{eq:Convex-Lemma-1}}		=
		\rho\,(1-\rho)\,
		\E^\Pk\left[ \int_0^T
		\left\{
		a_k \Delta_t^2 + \phi_k \left( q_t^{\Delta} \right)^2 + 2\Psi_k \Delta_t q_t^{\Delta}
		\right\} \,dt
		\right] \;.
	\end{equation}
As $\rho\in(0,1)$, we only need to demonstrate that the expected value is greater than zero. As $\phi_k \geq 0$,  the middle term in \eqref{eq:Convec-Lemma-2} is $\geq 0$. Next, let us focus on the right-most term in~\eqref{eq:Convec-Lemma-2}. Because $q_0^{\Delta} = 0$, we may write $q_t^{\Delta} = \int_0^t \Delta_u \, du$. Using integration by integrating by parts then yields
	\begin{equation}
		\E^\Pk\int_0^T 2\,\Delta_t\,q_t^{\Delta} \, dt
		=
		\E^\Pk\left[
		\left( q_T^\Delta \right)^2\right]
		\geq 0
		\;.
	\end{equation}
As $\Psi_k \geq 0$, this inequality implies the right-most term in \eqref{eq:Convec-Lemma-2} is non-negative. Lastly, notice that if $(\mathbb{P} \times \mu) ( \nu_t \neq \omega_t) > 0 $, then $(\Pk \times \mu) ( \nu_t \neq \omega_t) > 0 $ by absolute continuity of the measures, and therefore
	\begin{equation}
		\E^\Pk\left[ \int_0^T
		\Delta_t^2
		\; dt \right] >0
		\;.
	\end{equation}
As $a_k>0$, this result together with the inequality from the other two terms, shows that \eqref{eq:Convex-Lemma-1} is strictly greater than zero.
\end{proof}

\subsection{Proof of Lemma~\ref{thm:Differentiable-Lemma}}
\label{sec:Proof-Differentiable-Lemma}

\begin{proof}
	Using the definition of the G\^ateaux derivative,
	\begin{equation}
		\left\langle \D\Hbar^{\bnubar}_j(\nu) , \omega \right\rangle =
		\lim_{\epsilon\searrow0}
		\frac{\Hbar^{\bnubar}_j (\nu + \epsilon\,\omega) - \Hbar^{\bnubar}_j (\nu)}{\epsilon}
	\end{equation}
we  aim to show this limit exists and is equal to the result provided in the lemma. Using the representation for the objective $\Hbar^{\bnubar}_j$ in~\eqref{eq:Objective-Function-Alternative-Representation}, canceling out the $t=0$ terms, and using the linearity of the process $\qj{\nu}_t - \qj{\nu}_0$ in the variable $\nu$, we have
\begin{equation}
\begin{split}
	\Hbar^{\bnubar}_j(\nu+\epsilon\,\omega) - \Hbar^{\bnubar}_j(\nu)
	=& \quad
	\epsilon \,
	\E^\Pk\left[
	\int_0^T
	\left\{
	( \qj{\omega}_t - \qj{\omega}_0 ) (A_t^k + \blambda_k^\T \, \bnubar_t)
	- 2
	\left(\begin{smallmatrix} \nu_t \\ \qj{\nu}_t \end{smallmatrix}\right)^\T
	\bm\Gamma_k	
	\left(\begin{smallmatrix} \omega_t \\ \qj{\omega}_t - \qj{\omega}_0 \end{smallmatrix}\right)
	\right\} dt	\right]
	\\&-
	\epsilon^2\,
	\E^\Pk\left[
	\int_0^T
	\left(\begin{smallmatrix} \omega_t \\ \qj{\omega}_t - \qj{\omega}_0 \end{smallmatrix}\right)^\T
	\bm\Gamma_k
	\left(\begin{smallmatrix} \omega_t \\ \qj{\omega}_t - \qj{\omega}_0 \end{smallmatrix}\right)
	\, dt
	\right]
	\;,
\end{split}
\end{equation}
where $\bm\Gamma_k =
	\begin{pmatrix}
		a_k & \Psi_k \\ \Psi_k & \phi_k
	\end{pmatrix}
$.
Dividing by $\epsilon$ and taking the limit yields
\begin{equation} \label{eq:j-gateax-1}
	\left\langle \D \Hbar^{\bnubar}_j(\nu) , \omega \right\rangle
	=
	\E^\Pk\left[
	\int_0^T
	\left\{
	(\qj{\omega}_t - \qj{\omega}_0 )( A_t^k + \blambda_k^\T \, \bnubar_t )
	- 2
	\left(\begin{smallmatrix} \nu_t \\ \qj{\nu}_t \end{smallmatrix}\right)^\T
	\bm\Gamma_k	
	\left(\begin{smallmatrix} \omega_t \\ \qj{\omega}_t - \qj{\omega}_0 \end{smallmatrix}\right)
	\right\} \, dt
	\right]
	\;.
\end{equation}
Expanding the right part of the integrand in \eqref{eq:j-gateax-1} and re-grouping terms,
\begin{equation} \label{eq:j-gateax-2}
\begin{split}
	\left\langle \D \Hbar^{\bnubar}_j(\nu) , \omega \right\rangle
	=
	&\, \E^\Pk\left[\;
	\int_0^T
	( \qj{\omega}_t - \qj{\omega}_0 ) \left( A_t^k + \blambda_k^\T \, \bnubar_t - 2 ( \phi_k \qj{\nu}_t + \Psi_k \nu_t ) \right)
	 dt
\right.
\\
& \qquad \left.
-2	\int_0^T
	\omega_t \left(
	a_k \nu_t + \Psi_k \qj{\nu}_t
	\right)  dt
	\right].
\end{split}
\end{equation}

As $\nu,\omega\in\A^j$ and $\nubar,\Ahat\in\HT$, the sufficient conditions for Fubini's theorem are met. Applying Fubini's theorem, the tower property and the fact that $\omega_t$ is $\F_t^j$-measurable,

\begin{align*}
\left\langle \D \Hbar^{\bnubar}_j(\nu) , \omega \right\rangle	 &=
	\int_0^T
	\E^\Pk\left[
	\omega_t \, \left(
	-2a_k \nu_t - 2 \Psi_k \qj{\nu}_T
	+ \int_t^T \left\{  A_u^k + \blambda_k^\T \, \bnubar_u - 2\phi_k \qj{\nu}_u \right\} \,du  \right)
	\right] \, dt
	\\ &=
 	\int_0^T
	\E^\Pk\left[
	\omega_t \, \left(
	-2a_k \nu_t - 2 \Psi_k \qj{\nu}_T
	+ \Ek\left[
	\int_t^T \left\{  A_u^k + \blambda_k^\T \, \bnubar_u - 2\phi_k \qj{\nu}_u \right\} \,du
	\; \Bigg\lvert \F_t^j \right] \right)
	\right] \, dt
	\\ &=
 	\int_0^T
	\E^\Pk\left[
	\omega_t \, \left(
	-2a_k \nu_t - 2 \Psi_k \qj{\nu}_T
	+
	\int_t^T \Ek\left[   A_u^k + \blambda_k^\T \, \bnubar_u - 2\phi_k \qj{\nu}_u  \; \lvert \F_t^j \right] \,du
	 \right)
	\right] \, dt
	\\ &=
 	\int_0^T
	\E^\Pk\left[
	\omega_t \, \left(
	-2a_k \nu_t - 2 \Psi_k \qj{\nu}_T
	+
	\int_t^T \Ek\left[ \Ek\left[   A_u^k + \blambda_k^\T \, \bnubar_u \lvert \F_u^j \right] - 2\phi_k \qj{\nu}_u  \; \lvert \F_t^j \right] \,du
	 \right)
	\right] \, dt
	\\ &=
 	\int_0^T
	\E^\Pk\left[
	\omega_t \, \left(
	-2a_k \nu_t - 2 \Psi_k \qj{\nu}_T
	+
	\int_t^T \Ek\left[   A_u^k + \blambda_k^\T \, \bnubar_u - 2\phi_k \qj{\nu}_u  \; \lvert \F_t^j \right] \,du
	 \right)
	\right] \, dt
	\\ &=
 	\int_0^T
	\E^\Pk\left[
	\omega_t \, \left(
	-2a_k \nu_t - 2 \Psi_k \qj{\nu}_T
	+
	\int_t^T \left\{\Ek\left[   A_u^k + \blambda_k^\T \, \bnubar_u \lvert \F_u^j \right] - 2\phi_k \qj{\nu}_u   \,du\right\}
	 \right)
	\right] \, dt
\end{align*}

which gives the desired result.
\end{proof}

\subsection{Proof of Theorem~\ref{thm:FBSDE-Optimality-Condition}}
\label{sec:Proof-FBSDE-Optimality-Condition}

\begin{proof}
	By using lemmas~\ref{thm:Convex-Lemma} and~\ref{thm:Differentiable-Lemma} we may apply the results of \cite[Section 5]{ekeland1999convex} which state that, for each $j\in\mfJ$
	\begin{equation}
		\langle \D \Hbar_j^{\bnubar} (\nujst) , \omega \rangle = 0, \; \forall \omega\in\A^j \qquad \Leftrightarrow 	\qquad	\nujst = \argmax_{\nu\in\A^j} \Hbar^{\bnubar}_j(\nu)
		\;.
\end{equation}
Further, the strict concavity of $\Hbar$ implies that $\nujst$ is unique up to $\P\times\mu$ null sets. Therefore we need only demonstrate that $\langle \D \Hbar^{\bnubar}_j(\nujst) , \omega \rangle = 0$, $\forall \omega\in\A^j$, if and only $\nujst$ is the solution to the FBSDE~\eqref{eq:Optimality-Consistency-Condition}.

	\textsl{Sufficiency:} Suppose that $\nujst$ is the solution to the FBSDE~\eqref{eq:Optimality-Consistency-Condition} and that $\nujst\in\HT$. We now show that $\nujst\in\A^j$ and that $\langle \D \Hbar^{\bnubar}_j(\nujst) , \omega \rangle = 0$, $\forall \omega\in\A^j$.

First, the solution to the FBSDE may be represented implicitly as
	\begin{equation} \label{eq:proof-prop-optim-FBSDE-eq0}
		2 \, a_k \, \nujst_t =
\E^\Pk \left[
\left.
- 2 \,\Psi_k \,\qj{\nujst}_T
	+ \int_t^T \left\{
	\Ek\!\left[\,
	A_u^k + \blambda_k^\T\, \bnubar_u \,\lvert\, \F_u^j
	\,\right]
	- 2\phi_k\, \qj{\nujst}_u  \right\} \,du \,
\right \lvert \,\F_t^j \right]
	\;,
	\end{equation}
	which demonstrates that $\nujst$ is $\F^j$-adapted. Therefore, since $\nujst\in\HT$ and $\nujst$ is $\F^j$-adapted, we have that $\nujst \in \A^j$. Second, by inserting~\eqref{eq:proof-prop-optim-FBSDE-eq0} into the expression for the G\^ateaux derivative~\eqref{eq:Gateaux-Derivative-Expression} from Lemma~\ref{thm:Differentiable-Lemma} and using the tower property, we find that it vanishes almost surely.

	\textsl{Necessity:} Suppose that $\langle \D \Hbar^{\bnubar}_j(\nujst) , \omega \rangle = 0$, $\forall\;\omega\in\A^j$, then
\begin{equation}  \label{eq:Proof-Optimality-Eq-1}
\E^\Pk\left[
\left.
-2a_k \nujst_t - 2 \Psi_k \qj{\nujst}_T
+ \int_t^T \left\{ \Ek\left[ A^k_u + \blambda_k^\T\, \bnubar_u \,\lvert \, \F_u^j \right] - 2\phi_k \qj{\nujst}_u  \right\} \,du \,\right\lvert \,\F_t^j \right] = 0, \qquad\P\times\mu\;a.e.
\end{equation}
To see this, suppose that $\langle \D \Hbar^{\bnubar}_j(\nujst) , \omega \rangle = 0$ for all $\omega\in\A^j$, but \eqref{eq:Proof-Optimality-Eq-1} does not hold. Then, choose $\tomega=(\tomega_t)_\tT$ s.t.,
	\begin{equation} \label{eq:proof-prop-optim-FBSDE-eq2}
		\tomega_t = \E^\Pk \left[
\left.
	-2a_k \nujst_t - 2 \Psi_k \qj{\nujst}_T
	+ \int_t^T \left\{ \Ek\left[ A^k_u + \blambda_k^\T\, \bnubar_u \,\lvert \, \F_u^j \right]- 2\phi_k \qj{\nujst}_u  \right\} \,du \,\right\lvert \,\F_t^j \right]
	\;.
	\end{equation}
Such $\tomega$ is $\F^j$-adapted by its very definition. Second, as $\nubark$, $\nujst$, $A^k \in \HT$, Jensen's and the triangle inequality applied to \eqref{eq:proof-prop-optim-FBSDE-eq2} implies the bound
	\begin{align*}
		\E^\Pk \left[\int_0^T ( \tomega_t )^2 \, dt\right]
		&\leq
		C_k \left(
		\E^\Pk \left[\int_0^T ( \nujst_t )^2 \, dt\right]
		\, +
		\E^\Pk \left[\int_0^T \left( (A_t^k)^2 + \lambda^2 (\nubar_t)^2 \right) \, dt\right]\right)
<\infty
		\;,
	\end{align*}
where the constant $C_k=4 \left( 1+  a_k^2 + T \, \Psi_k^2 + T^2 \phi_k^2  \right)$. Hence, $\tomega\in \HT$ and therefore $\tomega\in\A^j$. Inserting this choice of $\tomega$ into the expression for the  G\^ateaux derivative~\eqref{eq:Gateaux-Derivative-Expression}, we see that $\langle \D\Hbar^{\bnubar}_j(\nujst) , \tomega \rangle > 0$, and hence contradicts the assumption that $\langle \D \Hbar^{\bnubar}_j(\nujst) , \omega \rangle = 0$, $\forall\omega\in\A^j$.

Thus, using \eqref{eq:Proof-Optimality-Eq-1} and noting that $\nujst_t$ is $\F^j$-adapted, using the tower property, we may write
	\begin{equation}
		2 \, a_k \, \nujst_t = \E^\Pk \left[\left. - 2 \,\Psi_k \,\qj{\nujst}_T
	+ \int_t^T \left\{ \E^\Pk \left[\left.A_u^k + \blambda_k^\T \, \bnubar_u \,\right\lvert \,\F_u^j \right] - 2\,\phi_k \,\qj{\nujst}_u  \right\} \,du \,\right\lvert \, \F_t^j \right]
	\;,
	\end{equation}
	and
	\begin{equation}
		2 \, a_k \, \mMbar^j_t = \E^\Pk \left[ \left. - 2 \,\Psi_k \,\qj{\nujst}_T
	+ \int_0^T \left\{\E^\Pk \left[\left.A_u^k + \blambda_k^\T \, \bnubar_u \,\right\lvert \F_u^j \right] - 2\,\phi_k \,\qj{\nujst}_u  \right\} \,du \,\right\lvert \,\F_t^j \right]
	\;,
	\end{equation}
	which solves the FBSDE in the statement of the proposition.
\end{proof}

\subsection{Proof of Theorem~\ref{thm:Solution-g1-g2-h2}}
\label{sec:Proof-Solution-g1-g2-h2}

We separate this proof in 3 parts, corresponding to each of the claims of the proposition.

\textbf{Part (I):}

To obtain the solution to $\bg_{1}$, we first compute the SDE for $\bmcE^{\Q} \bg_{1}$, using the SDE for $\bg_{1}$ in Equation~\eqref{eq:g1-BSDE} and the SDE of $\bmcE^{\Q}$ in Equation~\eqref{eq:proposition-mcE}. After expanding the SDE for $\bmcE^{\Q} \bg_{1}$, and grouping terms, we find
	\begin{equation} \label{eq:Proof-Prop-g1-g2-h2-1}
		-d\left( \bmcE^{\Q}_t \bg_{1,t} \right)
		=
		\bmcE^{\Q}_t \bAhat_t \, dt
		- \bmcE^{\Q}_t
		\left\{
	    d\bmcMbar_t - (\bZ_t^{\Q})^{-1} d\left[ \bZ^{\Q}, \bmcMbar \, \right]_t
	   + (\bZ_t^{\Q})^{-1} d\bZ_t^{\Q} \, \bmcE^{\Q}_t \bg_{1,t}
		\}
		\right\}
		\;.
	\end{equation}
As $\bZ^{\Q}$ is a Radon-Nikodym derivative process, it must be a $\Q$-martingale, and by extension, the term $\bmcE^{\Q}_t(\bZ_t^{\Q})^{-1} d\bZ_t^{\Q} \, \bmcE^{\Q}_t \bg_{1,t}$ is the increment of a $\Q$-martingale. Next, by the Girsanov-Meyer theorem \cite{protter2005stochastic}[Chapter III, Thm. 35], the remainder of the terms in the curly brackets of Equation~\eqref{eq:Proof-Prop-g1-g2-h2-1} sum to the increment of a $\Q$-martingale. Because of this, we may re-write the BSDE for $\bmcE^{\Q}_t \bg_{1,t}$ as
	\begin{equation}
		-d\left( \bmcE^{\Q}_t \bg_{1,t} \right)
		=
		\bmcE^{\Q}_t \bAhat_t \, dt
		- d\bmcMt_t
		\;,
	\end{equation}
	for some martingale term $\bmcMt$. Using this last result, we may write out the implicit form of the solution as
	\begin{equation}
		\bmcE_t^{\Q} \, \bg_{1,t} =
		\E^\Q \left[\left.
		\int_t^T \, \bmcE_u^{\Q} \, \bAhat_u \, du \,\right\lvert \,\F_t^j
		\right]
		\;.
	\end{equation}
	Lastly, multiplying the result on both sides by $(\bmcE_t^{\Q})^{-1}$, we obtain the stated solution.

\textbf{Part (II):}

The ODE~\eqref{eq:g2-ODE} is a matrix-valued non-symmetric Riccati-type ODE. We prove the claims concerning the ODE~\eqref{eq:g2-ODE} by applying theorems and tools for non-symmetric Riccati ODEs in~\cite{freiling2000non} and \cite{freiling2002survey}. Firstly, define $\bgt_{2,t}=\bg_{2,T-t}$.
We show that all of the claims hold for $\bgt_{2,t}$, and hence also for $\bg_{2,t}$.

From ODE~\eqref{eq:g2-ODE}
	\begin{equation} \label{eq:Proof-g2-solution-1}
\left\{
		\begin{array}{rl}
		\partial_t \bgt_{2,t} \!\!\!&=
		 \left( \bm{\Lambda} + \bgt_{2,t} \right) \left( 2 \ba\right)^{-1} \bgt_{2,t}
		- 2\bm{\phi}
		\\
		\bgt_{2,0} \!\!\!&= -2\bm{\Psi}
		\end{array}
\right.
	\;.
	\end{equation}
Next we aim to use Theorem~2.3 of \cite{freiling2000non} on $\tilde\bg_{2,t}$ to prove existence and boundedness of a solution. Using the notation in \cite{freiling2000non}, define
	\begin{equation}
		B_{11} = \bm{0} ,\; B_{12} = -J ,\;
		B_{21} = - 2\bm{\phi}\;, B_{22} = \bm{\Lambda} J\;,
	\end{equation}
	and $W_0 = -2\bm{\Psi}$, where $J = (2\ba)^{-1}$. To meet the requirements of Theorem~2.3 in \cite{freiling2000non}, we must find $C,D \in \mathbb{R}^{K\times K}$, $C=C^\T$ so that $L + L^\T \leq 0$ and $C + D W_0 + W_0^\T D^\T > 0$, where
	\begin{equation}
		L =
		\begin{pmatrix}
			- 2 D\bm{\phi} & -C J + D \bm{\Lambda} J \\
			0 & -J^\T D
		\end{pmatrix}
		\;.
	\end{equation}

	Let $D = \bm{I}^{(K\times K)}$ and $C = 5\bm{\Psi}$. With these choices of $C,D$, and using the fact that $\Psi$ is a diagonal matrix with positive entries, we find that
	\begin{equation}
		C + D W_0 + W_0^\T D^\T = \bm{\Psi} > 0
		\;,
	\end{equation}
	which meets one of the necessary conditions. The choices of $C$ and $D$ also imply that the matrix $L$ takes the form
	\begin{equation}
		L =
		\begin{pmatrix}
			- 2 \bm{\phi} & - (5 \bm{\Psi} + \bm{\Lambda}) J \\
			0 & -J
		\end{pmatrix}
		\;.
	\end{equation}
Next, as $\det(L) = \det(-2\bm{\phi}) \times \det(-J)$, the set of eigenvalues of $L$ is the union of the set of eigenvalues of $-2\bm{\phi}$ and those of $-J$. Because $-2\bm{\phi} \leq 0$ and $-J < 0$, all eigenvalues of $L$ are guaranteed to be non-positive, and at least one of them is guaranteed to be non-zero, implying that $L<0$. Hence, $L + L^\T < 0$ which meets the second condition of Theorem 2.3 of \cite{freiling2000non}, and guarantees the existence of a solution to the ODE~\eqref{eq:Proof-g2-solution-1} and hence of \eqref{eq:g2-ODE}.

As the solution to $\bm{g}_{2,t}$ exists and is continuous on the interval $[0,T]$, it follows that it is also bounded on this interval. Since the solution is guaranteed to exist and to be bounded, we may apply \cite[Thm 3.1]{freiling2002survey}, which guarantees that the solution is unique and takes the form~\eqref{eqn:proposition-solution-g2}, as desired.


\textbf{Part (III):}

The reader may verify that the presented solution for the Ricatti ODE~\eqref{eq:h2k-ODE-Equation} is valid. Moreover, it is also easy to verify that the solution is bounded and continuous in the interval $[0,T]$. All that remains is to show that $h_{2,t}^k \leq 0 $ for all $\tT$. If we notice that since $t<T$ and $\gamma_k\geq 0$ that $\sinh(-\gamma_k (T-t)) \leq 0$ and $\cosh(-\gamma_k (T-t)) \geq 1$. As $\xi_k, \Psi_k \geq 0$ we then get that
\begin{equation}
	\frac{
		\Psi_k \cosh\left( -\gamma_k (T-t)\ \right)
		- \xi_k \sinh\left( -\gamma_k (T-t)\ \right) }{
		 \xi_k \cosh\left( -\gamma_k (T-t)\ \right)
		- \Psi_k \sinh\left( -\gamma_k (T-t)\ \right)
		}
		\geq 0
		\;,
\end{equation}
and the desired result follows. \qed

\subsection{Proof of Theorem~\ref{thm:Final-Solution-Statement}}
\label{sec:Proof-Final-Solution-Statement}

To demonstrate the claim of the theorem, we need to show that the optimality conditions of Theorem~\eqref{thm:FBSDE-Optimality-Condition} are fulfilled. As demonstrated in Section~\ref{sec:Solving-The-Optimality-FBSDE}, if there exists solutions to the Ricatti-type ODEs~\eqref{eq:h2k-ODE-Equation} for $\{h_{2,t}^k\}_{k\in\mfK}$, a matrix-valued Ricatti-type ODE~\ref{eq:g2-ODE} for $\bg_{2,t}$ as well as the vector-valued BSDE~\eqref{eq:g1-BSDE} for $\bg_{1,t}$, then the solution to the optimality FBSDE~\eqref{eq:Optimal-FBSDE-System} follows the exact form presented in the statement of this theorem. In Theorem~\ref{thm:Solution-g1-g2-h2}, we showed that there exist solutions to these FBSDEs, and hence the solution to the optimality FBSDE of Theorem~\eqref{thm:FBSDE-Optimality-Condition} is solved.

All that remains to be shown is that the solution to the optimality FBSDE also belongs to an individual agent's set of admissible strategies, $\A^j$ and that the consistency conditions are met.

First, we show that $\nujst\in\A^j$. To do this, we must demonstrate that $\nujst$ is $\F^j$-predictable and contained in $\HT$. By the definition of $\bZ^\Q$ in equation~\eqref{eq:Def-ZtQ}, it is an $\F$-adapted process, and by extension $\bmcE_t^{\Q}$ must also be $\F$-predictable. Therefore, by the definition of the conditional expected value, the solution to $\bg_{1,t}$ presented in Theorem~\ref{thm:Solution-g1-g2-h2} must be $\F$-predictable, and hence the mean-field processes $\{\nubark\}_{k\in\mfK}$ must all be $\F$-predictable as well. Lastly, since $\nujst_t = \nubark_t + \frac{h_{2,t}^k}{2 a_k} ( \qj{\nujst}_t - \qbark{\nubark}_t )$ and since $\qbark{\nubark}_t$ is $\F$-predictable, and since $h_{2,t}^k$ is deterministic, we have that $\nujst_t$ must be $\F^j$-adapted.

Next, we must show that $\nujst\in\HT$. Noting that $d\bqbar^{\bnubarst}_t = \bnubarst_t \, dt = (\bg_{1,t} + \bg_{2,t} \bqbar^{\bnubarst}_t )\, dt$ and that $\bqbar_0^{\bnubarst} = (\mbar_k)_{k\in\mfK}=\bm{\mbar}$, we can solve for $\bqbar_t$ directly as
\begin{equation}
	\bqbar^{\bnubarst}_t = \bmcE\left(\int_0^t \bg_{2,s} \, ds\right) \, \bm{\mbar}
	+ \int_0^t \bmcE\left(\int_0^s \bg_{2,s} \, ds\right) \bg_{1,s} \, ds
	\;,
\end{equation}
where $\bmcE\left(\int_0^t \bg_{2,s} \, ds\right)$ is the solution to the time-ordered matrix exponential of $\bg_{2,s}$. Thus by Yonge's inequality and the boundedness of $\bg_{2,t}$,
\begin{align}
	\E^{\Pk} \int_0^T \lVert \bqbar^{\bnubarst}_u \rVert^2 \, du
	&\leq 2 \left(
	\nrm{\bmbar}^2 \,
	\int_0^T \OPnrm{\bmcE\left(\int_0^t \bg_{2,s} \, ds\right)}^2
	\, ds
	+
	T \,
	\int_0^T
	\OPnrm{\bmcE\left(\int_0^s \bg_{2,s} \, ds\right)}
	\nrm{\bg_{1,s}}
	\, ds
	\right)
	\\ &\leq
	C_0 + C_1 	
	\int_0^T
	\nrm{\bg_{1,s}}
	\, ds < \infty
	\;,
\end{align}
for some $C_0,C_1>0$, where $\lVert \cdot \rVert_2$ represents the $\ell^2$ operator norm, . Hence, $\bqbar^{\bnubarst}\in\HT$.

Next, using this last fact, if we compute the expected integrated squared norm of $\bnubar$ over $[0,T]$, we find that
\begin{align}
	\E^{\Pk} \int_0^T \lVert \bnubarst_u \rVert^2 \, du
	&=
	\E^{\Pk} \int_0^T \left\lVert
	\bg_{1,t} + \g_{2,t} \, \bqbar_u^{\bnubar}
	\right\rVert^2 \, du
	\\ &\leq 2
	\left(
	\E^{\Pk} \int_0^T \left\lVert
	\bg_{1,t}
	\right\rVert^2 \, du
	+
	\E^{\Pk} \int_0^T \lVert
	\g_{2,t} \rVert_{2}^2
	\,
	\lVert \bqbar_u^{\bnubar}
	\rVert^2 \, du
	\right)
	\\ &\leq
	C_2
	+
	C_3 \,
	\E^{\Pk} \int_0^T
	\lVert \bqbar_u^{\bnubar}
	\rVert^2 \, du
	< \infty
	\;,
\end{align}
for some constants $C_2,C_3>0$ and where in the third line of the inequality we use the fact that the function $\bg_{2,t}$ is bounded over the interval $[0,T]$ and the fact that $\bg_1\in\HT$ (as stated in the conditions of the theorem). Hence, $\bnubar\in\HT$.

Next, notice that
\begin{equation}
	\E^{\Pk} \int_0^T \lvert \nujst_u \rvert^2 \, du
	\leq
	2 \left(
	\E^{\Pk} \int_0^T \lvert \nubarkst_u \rvert^2 \, du
	+
	\E^{\Pk} \int_0^T \lvert \nujst_u - \nubarkst_u \rvert^2 \, du
	\right)
	\;.
\end{equation}
As $\nubarkst\in\HT$, the above demonstrates that it is sufficient to show that $\nujst_u - \nubarkst_u \in \HT$ to guarantee that $\nujst\in\HT$.

Similarly to $\bqbar_t$, if we notice that $d(\qj{\nujst}_t - \qbark{\nubarkst}_t) = (\nujst_t - \nubarkst_t) \, dt = \frac{h_{2,t}^k}{2 a_k} (\qj{\nujst}_t - \qbark{\nubarkst}_t) \, dt$ and that $(\qj{\nujst}_0 - \qbark{\nubarkst}_0)= \mfQ_0^j - \mbar_k$, we can solve exactly for this difference as
\begin{equation} \label{eq:Solution-Diff-Inventories}
	\qj{\nujst}_t - \qbark{\nubarkst}_t = \left( \mfQ_0^j - \mbar_k \right) e^{\int_0^t \frac{h_{2,t}^k}{2a_k}}
	\;.
\end{equation}
As $\E^\Pk (\mfQ_0^j)^2 < \infty$ and $h_{2,t}^k \leq 0$ it is easy to see that $\left(\qj{\nujst}_t - \qbark{\nubarkst}_t\right) \in \HT$.

Using the solution to $\nujst$ and using the result above,
\begin{equation}
	\E^{\Pk} \int_t^T \lvert \nujst_u - \nubarkst_u \rvert^2 \, du
	\leq
	\frac{\sup_{\tT } (h_{2,t}^k)^2}{4 a_k}
	\E^{\Pk} \int_t^T \left\lvert \qj{\nujst}_t - \qbark{\nubarkst}_t  \right\rvert^2 \, du
	<\infty
	\;,
\end{equation}
where we use $h_{2,t}^k < 0$ in line 3. Hence, $\nujst_u - \nubarkst_u \in \HT$ and $\nujst\in\HT$. Thus we have demonstrated that $\nujst$ is $\F^j$-predictable, and that $\nujst\in\HT$. Therefore $\nujst\in\HT$.

Lastly, we demonstrate that the consistency conditions are met. In other words, we must show that
\begin{equation}
	\nubarkst_t = \lim_{N\rightarrow\infty} \frac{1}{N_k^{\N}} \sum_{j\in\K_k^{\N}} \nujst_t
\end{equation}
for all $\tT$ and for all $k\in\mfK$. Using the solution to $\qj{\nujst}_t - \qbark{\nubarkst}_t$, we find that
\begin{equation}
	\lim_{N\rightarrow\infty} \frac{1}{N_k^{\N}} \sum_{j\in\K_k^{\N}} \left( \nujst_t - \nubarkst_t \right)
	 =
	e^{\int_0^t \frac{h_{2,t}^k}{2a_k}}
	\lim_{N\rightarrow\infty} \frac{1}{N_k^{\N}} \sum_{j\in\K_k^{\N}} \left( \mfQ_0^j - \mbar_k \right)
	\;.
\end{equation}
Now since the $\mfQ_0^j$ have bounded variance, the limit on the right vanishes as $N\rightarrow\infty$ by the law of large numbers. Hence, the consistency conditions are met.

The last statement follows from Theorem~\ref{thm:FBSDE-Optimality-Condition}.
\qed

\subsection{Proof of Proposition~\ref{prop:Z-Prior-Expression-Prop}}
\label{sec:Proof-Z-Prior-Expression-Prop}

\begin{proof}
	Let us first note that we may represent each element in $\bZ_t^{\Pk}$ as a Doob-martingale since
\begin{equation}
	\frac{d\P^{k^\prime}}{d\Pk} \Big \lvert_{\F_t}
	=
	\E^\Pk \left[
	\frac{d\P^{k^\prime}}{d\Pk} \;\Big \lvert_{\G_T}
	\Bigg\lvert \, \F_t \right]
	\;.
\end{equation}
Recall the global filtration $\mathfrak{G}=(\G_t)_{t\in[0,T]}$ introduced in Section~\ref{sec:The-State-Processes}, with the property that $\G_t \supseteq \bigvee_{j\in\mfN} \F_t^{j}$ for all $t\in[0,T]$. By this definition, we have that
\begin{equation}
	\bZ_t^{\Pk} =
	\text{diag}
	\left(
	\E^\Pk \left[
	\frac{d\P^{k^\prime}}{d\Pk} \Big \lvert_{\G_T}
	 \;
	\Bigg\lvert \, \F_t \right]
	\right)_{k^{\prime} \in \mfK}
	\;.
\end{equation}
Each term $\frac{d\P^{k^\prime}}{d\Pk} \Big \lvert_{\G_T}$ is in fact quite easy to compute. Let us remember that only difference between measures $d\Pk$ and $d\P^{k^\prime}$ is the law of the initial value of the latent process, $\Theta_0$. For each $k\in\mfK$, we have that $\Pk(\Theta_0 = \theta_j) = \pi_0^{k,j}$. Thus, we may write the expression for each Radon-Nikodym derivative conditional on $\G_T$ as
\begin{equation}
	\frac{d\P^{k^\prime}}{d\Pk} \Big \lvert_{\G_T} = \sum_{i\in\mfJ}
	\frac{\pi^{k^{\prime},i}_0}{\pi^{k,i}_0} \1{\Theta_0=\theta_i}
	\;.	
\end{equation}
As each $\frac{\pi^{k^{\prime},i}_0}{\pi^{k,i}_0}$ is constant, taking the conditional expected value with respect to $\Pk$ yields
\begin{equation}
	\frac{d\P^{k^\prime}}{d\Pk} \Big \lvert_{\F_t} = \sum_{i\in\mfJ}
	\frac{\pi^{k^{\prime},i}_0}{\pi^{k,i}_0} \, \Pk\left( \Theta_0=\theta_i \Big \lvert \F_t \right)
	\;.
\end{equation}
Assembling the $\frac{d\P^{k^\prime}}{d\Pk}$ terms above into a diagonal matrix, we find that the expression for $\bZ_t^{\Pk}$ follows the form in the statement of the proposition.
\end{proof}

\subsection{Proof of Proposition~\ref{prop:Example-Admissibility}}
\label{sec:Proof-Example-Admissibility}

\begin{proof}
	We will need to show here that the expression for $\bg_{1,t}$ presented in Theorem~\ref{thm:Solution-g1-g2-h2} satisfies $\bg\in\HT$. In other words, we need to show that $\E^{\Pk} \left[ \int_0^T \lVert \bg_{1,t} \rVert^2 \, dt \right] < \infty$ for all $k\in\mfK$.

	The first step will be to show that the operator norm of $\bmcE_t^{\Pk}$ is almost surely bounded above when using the latent Markov chain model. For the remainder of this proof, we suppress the superscript $\Pk$ for ease of notation. Simply applying It\^o's lemma, we find that $\bmcE_t = \bmcEt_t \bZ_t^{\Pk}$, where $\bmcEt_t$ is the solution to the SDE
	\begin{equation}
		d\bmcEt_t = \bmcEt_t \, \bZ_t^{\Pk} \bG_t \, (\bZ_t^{\Pk})^{-1} \, dt
	\end{equation}
	with the initial condition $\bmcEt_0 = \bm{I}^{K \times K}$. Writing out the implicit solution of the differential equation and taking the operator norm we find that
	\begin{align}
		\OPnrm{\,\bmcEt_t}
		&= \OPnrm{ \bm{I}^{K \times K} + \int_0^t \bmcEt_u \, \bZ^{\Pk}_u \bG_u (\bZ_u^{\Pk})^{-1} \, du }
		\\ &\leq
		1 + \int_0^t \OPnrm{ \,\bmcEt_u \, \bZ^{\Pk}_u \bG_u (\bZ^{\Pk}_u)^{-1} } \, du
		\\ &\leq
		1 + \int_0^t \OPnrm{ \,\bmcEt_u } \OPnrm{\bZ^{\Pk}_u} \OPnrm{\bG_u} \OPnrm{(\bZ_u^{\Pk})^{-1} } \, du
		\;,
	\end{align}
	where we use the triangle inequality, Jensen's inequality and the property of the operator norm. As shown in Proposition~\ref{prop:Z-Prior-Expression-Prop}, we know that $\bZ^{\Pk}_t$ is almost surely bounded over the interval $[0,T]$. From Theorem~\ref{thm:Solution-g1-g2-h2}, we also know that $\bG_t$ is bounded over this same interval. Now, looking back to the definition of $\bZ_t$, we find that
	\begin{equation}
		\bZ_t^{-1} = \text{diag} \left( \frac{d\P^{k^\prime}}{d\P^{k}} \Big \lvert_{\F_t} \right)
		\;,
	\end{equation}
	which can also be expressed in the same way as presented in Proposition~\ref{prop:Z-Prior-Expression-Prop}, which in turn implied that $\bZ_t$ is almost surely bounded over $[0,T]$. Therefore, it follows that there exists a constant $C_0 >0$ such that
	\begin{equation}
		\OPnrm{\,\bmcEt_t} \leq
		1 + C_0 \int_0^t \OPnrm{ \,\bmcEt_u } \, du
		\;.
	\end{equation}
	Applying Gr\"onwall's lemma to the above yields that $\sup_{\tT} \OPnrm{\,\bmcEt_t} \leq e^{C_0 T} < \infty$. Repeating the same analysis on $ \OPnrm{\bmcEt_t^{-1}}$ yields the very same bound. Finally, since the operator norms of $\bZ^{\Pk}_t$, $(\bZ^{\Pk}_t)^{-1}$, $\bmcEt_t$ and $\bmcEt_t^{-1}$ are all bounded over $[0,T]$, we get that there exists a constant $C_1>0$ such that $\sup_{t,u\in[0,T]} \OPnrm{(\bmcE_t)^{-1} \bmcE_u } < e^{TC_1}$

	Next, we wish to show that $\bAhat\in\HT$. Under our model, we may compute $\Ahat^k$ as
	\begin{align}
		\Ahat^k_t &= \E^{\Pk} \left[ \sum_{i=1}^J \alpha_t^{i} \, \1{\Theta_t = \theta_i} \big \lvert \F_t
		\right]
		\\ &=
		\sum_{i\in\mfJ} \alpha_t^i \, \Pk(\Theta_t = \theta_i \big \lvert \F_t )
		\;.
	\end{align}
	Therefore, since all of the $\Pk$ terms in the above are bounded above by 1, we may use Young's inequality to write
	\begin{align}
		\nrm{\bAhat\,}^2 \leq K^2 \sum_{i\in\mfJ} \nrm{\alpha_t^i}^2
		\;.
	\end{align}
As each $\alpha_t^i\in\HT$, we get that $\bAhat\in\HT$.

	Now we can proceed to showing the main result. Using the bounds we derived above and Jensen's inequality, we may write
	\begin{align}
		\E^{\Pk} \left[ \int_0^T \lVert \bg_{1,t} \rVert^2 \, dt \right]
		&\leq
		\E^{\Pk} \left[ \int_0^T \left\lVert
		\E^{\Pk} \left[
		\int_t^T \,
		(\bmcE_t)^{-1} \bmcE_u \, \bAhat_u
		\, du
		\Big\lvert \F_t \right]
		\right\rVert^2 \, dt \right]
		\\&\leq
		\E \left[ \int_0^T
		\int_t^T \,
		\left\lVert
		(\bmcE_t)^{-1} \bmcE_u \, \bAhat_u
		\right\rVert^2
		\, du
		 \, dt  \right]
		\\&\leq
		\E \left[ \int_0^T
		\int_t^T \,
		\OPnrm{(\bmcE_t)^{-1} \bmcE_u}^2 \, \left\lVert \bAhat_u
		\right\rVert^2
		\, du
		 \, dt  \right]
		\\&\leq
		(T+1)^2 \, e^{2 C_1 T} \,
		\E^{\Pk} \left[ \int_0^T
		\left\lVert \bAhat_u
		\right\rVert^2
		\, du  \right] < \infty
		\;,
	\end{align}
	where in the last line, we use the fact that $\bAhat\in\HT$. Thus, we find that $\bg_1\in\HT$, which verifies the claim of the proposition.
\end{proof}

\subsection{Proof of Theorem~\ref{thm:Epsilon-Nash-Theorem}}
\label{sec:Proof-Epsilon-Nash-Theorem}

We begin the proof of Theorem~\ref{thm:Epsilon-Nash-Theorem} by introducing a lemma regarding the distance between the mean-field game objective $\Hbar_j$ and the finite player game objective $H_j$.	
\begin{lemma} \label{lemma:Obj-Diff-Bound}
	Let $\nu \in \A^j$ be some arbitrary admissible control and $\nu^{-j,\ast}\in\A^{-j}$ be the collection $\nu^{-j,\ast}:=\left( \nu^{1,\ast} , \dots , \nu^{j-1,\ast},\nu^{j+1,\ast}, \dots, \nu^{N,\ast} \right)$ of optimal controls defined by equation~\eqref{eq:Solution-Structure-Statement} in Theorem~\ref{thm:Final-Solution-Statement} for all agents except for $j$. Let us also assume that $\bnubarst = \left(\nubarkst \right)_{k\in\mfK}$ follows the dynamics of equation~\eqref{eqn:meanfield-opt-trading-rate} in Theorem~\ref{thm:Final-Solution-Statement}. Then
	\begin{equation}
		\left\lvert H_j(\nu,\nu^{-j,\ast}) - \Hbar^{\bnubarst}_j(\nu) \right\rvert
		= o(\delta_N) + o(\frac{1}{N})
		\;.
	\end{equation}
\end{lemma}	
\begin{proof}
	 Using the definitions of $\Hbar^{\bnubarst}_j$ and $H_j$ and simplifying down the equations, we find that
	\begin{align}
		\left\lvert H_j(\nu,\nu^{-j,\ast}) - \Hbar_j(\nu) \right\rvert
		&=
		\left\lvert
		\E^{\Pk} \left[ \sum_{\kp\in\mfK} \int_0^T \, \lambda_{k,\kp}\left(
		p_{\kp}^{\N}\nubar^{\kp,\N}_t - p_{\kp} \nubar^{\kp}_t
		\right) \, dt \right]
		\right\rvert
		\\ &\leq
		\sum_{\kp\in\mfK}
		\lambda_{k,\kp} \,
		\left\lvert
		\E^{\Pk} \left[
		 \int_0^T \,
		p_{\kp}^{\N}\nubar^{\kp,\N}_t - p_{\kp} \nubar^{\kp,\ast}_t
		 \, dt  \right] \right\rvert
		 \label{eq:lemma-obj-diff}
	\end{align}
	Therefore it is sufficient for us to show that each of the expected values in the sum of \eqref{eq:lemma-obj-diff} is $o(N^{-1}) + o(\delta_N)$.

	Next, notice that using the definitions of $\nubar^{\kp,\N}_t$ and $p_{\kp}^{\N}$, we can decompose the difference of the mean-field rates between the agent's rate and the rate of all others
	\begin{align}
		p_{\kp}^{\N}\nubar^{\kp,\N}_t - p_{\kp} \nubar^{\kp,\ast}_t
		&= \frac{1}{N} (\nu_t - \nujst_t) + \frac{1}{N_k^{\N}} \sum_{i\in\K_{\kp}^{\N}} ( p_{\kp}^{\N} \nu^{j,\ast}_t - p_{\kp} \nubar^{\kp,\ast}_t) 		
		\;,
	\end{align}
	where $\nujst_t$ is the optimal control that agent-j would have taken in the limiting game.

	Using the triangle inequality and Jensen's along with the last result, we get that
{	\scriptsize
	\begin{equation}
		\text{\eqref{eq:lemma-obj-diff}}
		\leq
		\sum_{\kp\in\mfK}
		\lambda_{k,\kp} \,
		\left(
		\frac{1}{N}
		\E^{\Pk} \left[
		 \int_0^T \,
		\lvert \nu_t - \nujst_t \rvert
		 \, dt  \right]
		 +
		 \left\lvert
		 \E^{\Pk} \left[
		 \frac{1}{N_k^{\N}}
		\sum_{i\in\K_{\kp}^{\N}}
		\int_0^T \,
		(p_{\kp}^{\N} \nu^{j,\ast}_t - p_{\kp} \nubar^{\kp,\ast}_t)
		 \, dt  \right]
		 \right\rvert
		 \right)
		 \label{eq:lemma-obj-diff-2}
		\;.
	\end{equation}
}
	It is clear that $\nu_t - \nujst_t\in\A^j$ so we can guarantee that $E^{\Pk} \left[ \int_0^T \, \lvert \nu_t - \nujst_t \rvert \, dt  \right]$ is bounded and independent of $N$. Therefore,
	\begin{equation}
	 	\frac{1}{N} \,
		\E^{\Pk} \left[
		 \int_0^T \,
		\lvert \nu_t - \nujst_t \rvert
		 \, dt  \right]
		 = o(\frac{1}{N})
		 \;.
	\end{equation}
	Therefore all that's left to show is that the right part of the summand of \eqref{eq:lemma-obj-diff-2} vanishes at an appropriate speed.

	By plugging in the manipulation
	\begin{equation}
		p_{\kp}^{\N} \nu^{j,\ast}_t - p_{\kp} \nubar^{\kp,\ast}_t
		=
		(p_{\kp}^{\N} - p_{\kp}) \nu^{j,\ast}_t + p_{\kp} (\nu^{j,\ast}_t - \nubar^{\kp,\ast}_t)
	\end{equation}
	and using the triangle inequality and Jensen's inequality, we find that
	\begin{align}
		\left\lvert
		 \E^{\Pk} \left[
		 \frac{1}{N_k^{\N}}
		\sum_{i\in\K_{\kp}^{\N}}
		\int_0^T \,
		(p_{\kp}^{\N} \nu^{j,\ast}_t - p_{\kp} \nubar^{\kp,\ast}_t)
		 \, dt  \right]
		 \right\rvert
		 & \leq
		 \left\lvert
		 p_{\kp}^{\N} - p_{\kp}
		 \right\rvert
		 \E^{\Pk} \left[
		\int_0^T \,
		\lvert
		\nu^{j,\ast}_t
		\rvert
		 \, dt  \right] \label{eq:lemma-obj-diff-3}
		 \\&+
		 p_{\kp}
		 \left\lvert
		 \E^{\Pk} \left[
		 \frac{1}{N_k^{\N}}
		\sum_{i\in\K_{\kp}^{\N}}
		\int_0^T \,
		(\nu^{j,\ast}_t - \nubar^{\kp,\ast}_t)
		 \, dt  \right]
		 \right\rvert \label{eq:lemma-obj-diff-4}
	\end{align}

As $\nujst\in\A^j$, we find that $\E^{\Pk} \left[\int_0^T \, \lvert \nu^{j,\ast}_t \rvert \, dt  \right] < \infty$. Therefore by the assumption of the theorem, we get
\begin{equation}
		 \left\lvert
		 p_{\kp}^{\N} - p_{\kp}
		 \right\rvert
		 \E^{\Pk} \left[
		\int_0^T \,
		\lvert
		\nu^{j,\ast}_t
		\rvert
		 \, dt  \right]
		 = o(\delta_N)
		 \;.
\end{equation}
Next, using the structure of the solution for $\nujst_t$ from Theorem~\ref{thm:Final-Solution-Statement}, equation~\eqref{eq:Solution-Diff-Inventories} and the fact that $h_{2,t}^k$ is bounded, we get
\begin{align}
	\text{\eqref{eq:lemma-obj-diff-4}}
	&=
	p_{\kp}
	\left\lvert
 	\frac{1}{N_k^{\N}}
	\sum_{i\in\K_{\kp}^{\N}}
	\E^{\Pk} \left[
	\left( \mfQ_0^j - \mbar_k \right)
	\int_0^T \,
	e^{\int_0^t \frac{h_{2,t}^k}{2a_k}}
 	\, dt  \right]
 	\right\rvert
 	\\ &\leq
 	C_0
	\left\lvert
 	\frac{1}{N_k^{\N}}
	\sum_{i\in\K_{\kp}^{\N}}
	\left( \E^{\Pk} \left[ \mfQ_0^j \right] - \mbar_k \right)
 	\right\rvert
 	= 0
 	\;.
\end{align}
Hence, the right part of equation~\ref{eq:lemma-obj-diff-2} is equal to $o(\delta_N) + o(N^{-1})$ and the claims of the lemma hold true
\end{proof}

\subsubsection{Main Proof of Theorem~\ref{thm:Epsilon-Nash-Theorem}}
\begin{proof}
	We prove the result of the theorem by using the Lemma~\ref{lemma:Obj-Diff-Bound}. First, let us note that by the definition of the supremum,
	\begin{equation}
		H_j(\omega,\nu^{-j,\ast}) \leq \sup_{\nu\in\A^j} H_j(\nu,\nu^{-j,\ast})
	\end{equation}
	holds for all $\omega\in\A^{j}$, and therefore the left-most inequality in the statement of Theorem~\ref{thm:Epsilon-Nash-Theorem} holds.

	Next, we must show that the right-most inequality in the statement of Theorem~\ref{thm:Epsilon-Nash-Theorem} also holds. First let us note that by Lemma~\ref{lemma:Obj-Diff-Bound}, for any $\nu\in\A^j$,
	\begin{align}
		H_j(\nu,\nu^{-j,\ast})
		&\leq
		\Hbar_j^{\bnubarst}(\nu) + o(\delta_N) + o(N^{-1})
		\\ &\leq
		\Hbar_j^{\bnubarst}(\nujst) + o(\delta_N) + o(N^{-1})
		\;,
	\end{align}
	 where we use the fact that $\Hbar_j(\nujst) = \sup_{\nu\in\A^j} \Hbar_j(\nu)$. Applying Lemma~\ref{lemma:Obj-Diff-Bound} again, we find that
	\begin{equation}
			H_j(\nu,\nu^{-j,\ast}) \leq H_j(\nuj,\nu^{-j,\ast}) + 2 \, o(\delta_N) + 2 \, o(N^{-1})
			\;.
	\end{equation}
As the above inequality holds for all $\nu\in\A^j$ we may take the supremum on the left, and cancel out the constant terms multiplying the little-$o$ terms to yield the final result,
	\begin{equation}
		\sup_{\nu\in\A^j} H_j(\nu,\nu^{-j,\ast}) \leq H_j(\nuj,\nu^{-j,\ast}) + o(\delta_N) + o(N^{-1})
		\;.
	\end{equation}
\end{proof}

\section{Filtering and Smoothing Equations}
\label{sec:Filtering-and-Smoothing-Appendix}

In sections~\ref{sec:Example-Model-Subsection},~\ref{sec:Computational-Method} and~\ref{sec:Numerical-Experiments} we refer to the Radon-Nikodym process $\bZ^{\Q}$ and for the $\F$-projected drift process $\bAhat$ which are required for the approximation of the optimal control. This appendix will provide the details on how these quantities are computed for the mean-revertingmodel used in the numerical experiments present in Section~\ref{sec:Numerical-Experiments}.

Let us recall the model provided in Section~\ref{sec:Numerical-Experiments}. We assume that the un-impacted asset price process has the dynamics
\begin{equation*}
	dF_t = \kappa \left( \Theta_t - F_t \right) \, dt + \sigma \,dW_t\,,
\end{equation*}
where $\Theta_t$ is a continuous-time Markov chain with generator matrix $\bC$ which takes values in the set $\{\theta_i\}_{i=1}^{J}$. What varies across each measure $\Pk$ is the distribution over the initial state, $\Theta_0$, where we assume that $\Pk\left(\Theta_0=\theta_i \right) = \pi_0^{k,i}$ for each $i\in\{1, 2, \dots , J\}$ and $k\in\mfK$. Our first step will be to compute the $\F_t$-adapted process $\bAhat_t = \left( \Ek\left[ A_t \lvert \F_t \right] \right)_{k\in\mfK}$. Using the dynamics of $F_t$, we get that
\begin{align*}
	\Ek\left[ A_t \lvert \F_t \right]
	&=
	\Ek\left[ \kappa \left( \Theta_t - F_t \right) \lvert \F_t \right]
	\\ &=
	\kappa \left( \Ek\left[ \Theta_t \lvert \F_t \right] - F_t \right)
	\\ &=
	\kappa \left( \sum_{i=1}^J \theta_i \, \Pk\left( \Theta_t = \theta_i \big\lvert \F_t \right) - F_t \right)
	\;.
\end{align*}
Therefore, to compute $\bAhat$ we need to compute the posterior probabilities of each state of $\Theta_t$, $\Pk\left( \Theta_t = \theta_i \big\lvert \F_t \right)$. The lemma that follows gives an explicit way of computing these probabilities.
\begin{lemma}[Filtering Equation]
	Let us assume that the Novikov condition
  \begin{equation} \label{eq:NovikovCondition}
  \Ek\left[ \exp\left\{\int_0^T \left( A_u \right)^2
   \;du \right\} \right] < \infty
  \;
  \end{equation}
  holds for all $k\in\mfK$. For each $i=1,\dots,J$ and $k\in\mfK$, let $\pi_t^{k,i} = \Pk\left( \Theta_t = \theta_i \big\lvert \F_t \right)$, and define the processes $\Lambda^{k,i} = \left( \Lambda_t^{k,i} \right)_{t\in[0,T]}$,
	satisfying the dynamics
	\begin{equation*}
		d\Lambda_t^{k,i} =
		\Lambda_t^{k,i} \sigma^{-2} \kappa \left(\theta_i - F_t \right) dF_t
		+
		\sum_{j=1}^J C_{i,j} \Lambda_{t}^{k,j} \, dt
		\;,
	\end{equation*}
	along with the initial condition $\Lambda_0^{k,i} = \pi_0^{k,i}$. Then the filters $\pi_t^{k,j}$ satisfy the relation
	\begin{equation*}
		\pi_t^{k,j} = \Lambda_t^{k,i} \left/ \left(\sum_{j=1}^{J} \Lambda_t^{k,j} \right) \right.
		\;.
	\end{equation*}
\end{lemma}
\begin{proof}
	For the proof of this lemma, we refer the reader to the proof of a more general version of this statement found in \cite[Theorem 3.1]{casgrain_jaimungal_2016}.
\end{proof}

The next task is to compute the process $\bZ^{\Q}$ for any choice of $\Q=\Pk$ $k\in\mfK$. We can do this by applying Proposition~\ref{prop:Z-Prior-Expression-Prop} to the model dynamics that we have. This Proposition~\ref{prop:Z-Prior-Expression-Prop} allows us to compute $\bZ^{\Pk}$, given that we can compute the value of the time-0 smoothers for $\Theta$, $\Pk\left( \Theta_0 = \theta_i \big\lvert \F_t \right)$. The following lemma provides an expression for the computation of these smoothers.

\begin{lemma}[Smoothing Equation]
	Assume that the Novikov condition~\eqref{eq:NovikovCondition} holds. For each $k\in\mfK$ and $i,j\in\{1,2,\dots,J\}$, let us define the process $\Lambdat^{k,i,j} = \left( \Lambdat_t^{k,i,j} \right)_{t\in[0,T]}$,
	where each $\Lambdat_0^{k,i,j}$ satisfies the SDE
	\begin{equation*}
		d\Lambdat_t^{k,i,j} =
		\Lambdat_t^{k,i,j} \sigma^{-2} \kappa \left(\theta_j - F_t \right) dF_t
		+
		\sum_{\ell=1}^J C_{j,\ell} \Lambdat_{t}^{k,i,\ell} \, dt
		\;,
	\end{equation*}		
	and the initial condition $\Lambdat_0^{k,i,j} = \1{i=j}$. Then the time-0 smoother for $\Theta_0$ satisfies the equation
	\begin{equation*}
		\Pk\left( \Theta_0 = \theta_i \big\lvert \F_t \right)
		=
		\left( \sum_{j=1}^{J} \pi_0^{k,i} \Lambdat_t^{k,i,j} \right) \left/
		\left( \sum_{i,\ell=1}^{J} \pi_0^{k,i} \Lambdat_t^{k,i,\ell} \right)
		\right.
		\;,
	\end{equation*}
\end{lemma}
\begin{proof}
	For each $k\in\mfK$, let us define the measure $\tilde{\Q}^k$ which is specified through the Radon-Nikodym derivative
	\begin{equation*}
		\zeta_t^{k}
		= \frac{d\Pk}{ d\tilde{\Q}^k} \Big\lvert_{\F_t}
		=
		\exp\left\{
		\int_0^t A_u \, \sigma^{-2} \, dF_u - \frac{1}{2} \int_0^t (A_u)^2 \, \sigma^{-2} \, du
		\right\}
		\;.
	\end{equation*}	
	The Radon-Nikodym derivative above is defined specifically so that under measure $\tilde{\Q}^k$, $\left(F_t - F_0 \right) \sigma^{-1}$ is a Brownian motion, independent of $\Theta_t$ and so that the dynamics of $\Theta_t$ are left unchanged.

	Using this new measure, we can re-represent the time-0 smoother we are looking for as
	\begin{align*}
		\Pk\left( \Theta_0 = \theta_i \big\lvert \F_t \right)
		&=
		\frac{\E^{\tilde{\Q}^k} \left[ \1{\Theta_0 = \theta_i} \zeta_t^k \big\lvert \F_t \right]}{\E^{\tilde{\Q}^k} \left[ \zeta_t^k \big\lvert \F_t \right]}
		\\ &=
		\frac{\E^{\tilde{\Q}^k} \left[ \1{\Theta_0 = \theta_i} \zeta_t^k \big\lvert \F_t \right]}{\sum_{j=1}^J
		\E^{\tilde{\Q}^k} \left[ \1{\Theta_0 = \theta_j} \zeta_t^k \big\lvert \F_t \right]}
	\end{align*}
	Now, if we take a look at the term in the numerator, we can further expand it as
	\begin{align*}
		\E^{\tilde{\Q}^k} \left[ \1{\Theta_0 = \theta_i} \zeta_t^k \big\lvert \F_t \right]
		&=
		\sum_{j=1}^{J}
		\E^{\tilde{\Q}^k} \left[ \1{\Theta_0 = \theta_i} \1{\Theta_t = \theta_j} \zeta_t^k \big\lvert \F_t \right]
		\\ &=
		\pi_0^{k,i} \;
		\sum_{j=1}^{J}
		\E^{\tilde{\Q}^k} \left[ \1{\Theta_t = \theta_j} \zeta_t^k \big\lvert \F_t \vee \sigma\left(\Theta_0 = \theta_i \right) \right]
		\;,
	\end{align*}
	where we use Bayes' rule to get to the last line.

	Following the proof of \cite[Theorem 3.1]{casgrain_jaimungal_2016}, we find that
	\begin{equation*}
		\Lambdat_t^{k,i,j} =
		\E^{\tilde{\Q}^k} \left[ \1{\Theta_t = \theta_j} \zeta_t^k \big\lvert \F_t \vee \sigma\left(\Theta_0 = \theta_i \right) \right]
	\end{equation*}
	satisfies the SDE found in the statement of the theorem, with the initial condition $\Lambdat_0^{k,i,j} = \1{i=j}$. Plugging this back into the previous expressions, we obtain the final result.
\end{proof}

\clearpage
\bibliographystyle{chicago}
\bibliography{StochasticGamesPartialBib}

\end{document}